\documentclass[a4paper,USenglish]{lipics-v2019}
\pdfoutput=1

\usepackage[utf8]{inputenc}
\usepackage{xspace}
\usepackage{etoolbox}
\usepackage[autostyle]{csquotes}
\usepackage{microtype}
\usepackage{tikz}
\usetikzlibrary{shapes,fit,positioning,calc,backgrounds}

\theoremstyle{plain}
\newtheorem{observation}[theorem]{Observation}
\newtheorem*{theorem*}{Theorem}

\newcommand{\setof}[2]{\left\{#1 : #2 \right\}}
\newcommand{\Setof}[2]{\bigl\{#1 : #2 \bigr\}}
\newcommand{\ol}[1]{\overline{#1}}
\DeclareMathOperator{\vc}{{\rm vc}}
\DeclareMathOperator{\mw}{{\rm mw}}

\DeclareMathOperator{\tc}{{\rm tc}}
\DeclareMathOperator{\cvdn}{{\rm cvdn}}
\DeclareMathOperator{\nd}{{\rm nd}}
\DeclareMathOperator{\pw}{{\rm pw}}
\DeclareMathOperator{\tw}{{\rm tw}}
\DeclareMathOperator{\cw}{{\rm cw}}
\DeclareMathOperator{\poly}{{\rm poly}}
\DeclareMathOperator{\defic}{\mathrm{def}}

\newcommand{\N}{\mathbb{N}}

\newcommand{\BigO}{\mathcal{O}}

\newcommand{\FPT}{{\sf FPT}\xspace}
\newcommand{\NP}{{\sf NP}\xspace}
\newcommand{\W}[1]{{\sf W[#1]}\xspace}
\newcommand{\Wh}[1]{{\sf W[#1]\mbox{-}\nobreak\hspace{0pt}hard}\xspace}

\newcommand{\XP}{{\sf XP}\xspace}

\newcommand{\up}{{\rm pos}}
\newcommand{\down}{{\rm neg}}
\newcommand{\guard}{{\rm guard}}

\usepackage{framed}
\usepackage{tabularx}

\newlength{\RoundedBoxWidth}
\newsavebox{\GrayRoundedBox}
\newenvironment{GrayBox}[1]%
   {\setlength{\RoundedBoxWidth}{.93\columnwidth}
    \def\boxheading{#1}
    \begin{lrbox}{\GrayRoundedBox}
       \begin{minipage}{\RoundedBoxWidth}}%
   {   \end{minipage}
    \end{lrbox}
    \begin{center}
    \begin{tikzpicture}%
       \node(Text)[draw=black!20,fill=white,rounded corners,inner sep=2ex,text width=\RoundedBoxWidth]
             {\usebox{\GrayRoundedBox}};
        \coordinate(x) at (current bounding box.north west);
        \node [draw=white,rectangle,inner sep=3pt,anchor=north west,fill=white]
        at ($(x)+(6pt,.75em)$) {\boxheading};
    \end{tikzpicture}
    \end{center}}

\newenvironment{defproblemx}[1]{\noindent\ignorespaces%
                                \FrameSep=6pt%
                                \parindent=0pt%
                \begin{GrayBox}{#1}%
                \begin{tabular*}{\columnwidth}{!{\extracolsep{\fill}}@{\hspace{.1em}} >{\itshape} p{1.5cm} p{0.82\columnwidth} @{}}%
            }{
                \end{tabular*}%
                \end{GrayBox}%
                \ignorespacesafterend
            }

\newcommand{\prob}[3]{%
  \begin{defproblemx}{#1}
    Input: & #2 \\
    Task: & #3
  \end{defproblemx}
}

\newcommand{\TSS}{\textsc{Target Set Selection}\xspace}
\newcommand{\MAJTSS}{\textsc{Majority Target Set Selection}\xspace}
\newcommand{\UNITSS}{\textsc{Uniform Target Set Selection}\xspace}

\newcommand{\indeg}{\ensuremath{\deg^+}}
\newcommand{\outdeg}{\ensuremath{\deg^-}}

\title{Target Set Selection in Dense Graph Classes}

\author{Pavel Dvořák}{Computer Science Institute, Charles University, Prague, Czech Republic}{koblich@iuuk.mff.cuni.cz}{}{The research leading to these results has received funding from the European Research Council under the European Union's Seventh Framework Programme (FP/2007-2013) / ERC Grant Agreement n.~616787.}%mandatory, please use full name; only 1 author per \author macro; first two parameters are mandatory, other parameters can be empty.

\author{Dušan Knop}{Department of Theoretical Computer Science, Faculty of Information Technology, \\ Czech Technical University in Prague, Prague, Czech Republic}{dusan.knop@fit.cvutcz}{}{Research partly supported by the OP VVV MEYS funded project \\CZ.02.1.01/0.0/0.0/16\_019/0000765 ``Research Center for Informatics''.}

\author{Tomáš Toufar}{Computer Science Institute, Charles University, Prague, Czech Republic}{toufi@iuuk.mff.cuni.cz}{}{}

\authorrunning{P. Dvořák, D. Knop, and T. Toufar}%mandatory. First: Use abbreviated first/middle names. Second (only in severe cases): Use first author plus 'et. al.'

\Copyright{P. Dvořák, D. Knop, and T. Toufar}%mandatory, please use full first names. LIPIcs license is "CC-BY";  http://creativecommons.org/licenses/by/3.0/

%mandatory

\category{}%optional, e.g. invited paper

\relatedversion{}%optional, e.g. full version hosted on arXiv, HAL, or other respository/website

\supplement{}%optional, e.g. related research data, source code, ... hosted on a repository like zenodo, figshare, GitHub, ...

\funding{}%optional, to capture a funding statement, which applies to all authors. Please enter author specific funding statements as fifth argument of the \author macro.

\acknowledgements{}%optional

\EventEditors{John Q. Open and Joan R. Access}
\EventNoEds{2}
\EventLongTitle{42nd Conference on Very Important Topics (CVIT 2016)}
\EventShortTitle{CVIT 2016}
\EventAcronym{CVIT}
\EventYear{2016}
\EventDate{December 24--27, 2016}
\EventLocation{Little Whinging, United Kingdom}
\EventLogo{}
\SeriesVolume{42}
\ArticleNo{23}
\nolinenumbers %uncomment to disable line numbering
\hideLIPIcs  %uncomment to remove references to LIPIcs series (logo, DOI, ...), e.g. when preparing a pre-final version to be uploaded to arXiv or another public repository
\begin{document}
\maketitle
\begin{abstract}
  In this paper we study the \TSS problem from a parameterized complexity perspective.
  Here for a given graph and a threshold for each vertex the task is to find a set of vertices (called a target set) which activates the whole graph during the following iterative process.
  A vertex outside the active set becomes active if the number of so far activated vertices in its neighborhood is at least its threshold.

  We give two parameterized algorithms for a special case where each vertex has the threshold set to the half of its neighbors (the so-called \MAJTSS problem) for parameterizations by the neighborhood diversity and the twin cover number of the input graph.

  We complement these results from the negative side.
  We give a hardness proof for the \MAJTSS problem when parameterized by (a restriction of) the modular-width -- a natural generalization of both previous structural parameters.
  We also show the \TSS problem parameterized by the neighborhood diversity or by the twin cover number is \Wh{1} when there is no restriction on the thresholds.
\end{abstract}

\sloppy

\section{Introduction}
We study the \TSS problem (also called \textsc{Dynamic Monopolies}), using a notation according to Kempe et al.~\cite{KempeKT03}, from parameterized complexity perspective.
Let $G=(V,E)$ be a graph, $S\subseteq V$, and $f\colon V\to\N$ be a \emph{threshold function}.
The \emph{activation process} arising from the set $S_0 = S$ is an iterative process with resulting sets $S_0,S_1,\ldots$ such that for $i\ge 0$
\[
S_{i+1} = S_i \cup \setof{v\in V}{|N(v)\cap S_i|\ge f(v)},
\]
where by $N(v)$ we denote the set of vertices adjacent to $v$.
Note that after at most $n = |V|$ rounds the activation process has to stabilize -- that is, $S_n = S_{n+i}$ for all $i > 0$.
We say that the set $S$ is a \emph{target set} if $S_n = V$ (for the activation process $S = S_0,\dots,S_n$).

\prob{\TSS}
{A graph $G = (V, E)$, $f\colon V\to\N$, and a positive integer $b\in\N$.}
{Find a target set $S\subseteq V$ of size at most $b$ or report that there is no such set.}

We call the input integer $b$ the \emph{budget}.
The problem interpretation and computational complexity clearly may vary depending on the input function $f$.
There are three important settings studied -- namely constant, majority, and a general function.
If the threshold function~$f$ is the majority (i.e., $f(u) = \lceil \deg(u) / 2 \rceil$ for every vertex $u \in V$), we denote the problem as \MAJTSS.

\subparagraph{Motivation}
The \TSS problem was introduced by Domingos and Richardson~\cite{DomingosR01} in order to study influence of direct marketing on a social network.
It is noted therein that it captures e.g. \emph{viral marketing}~\cite{RichardsonD02}.
The \TSS problem is important also from the graph theoretic viewpoint, since it generalizes many well known \NP-hard problems on graphs.
These problems include
\begin{itemize}
  \item {\sc Vertex Cover}~\cite{Chen09} -- set $f(v) = \deg(v)$ for all $v\in V$.
  \item {\sc Irreversible $k$-Conversion Set}~\cite{DreyerR09}, {\sc $k$-Neighborhood Bootstrap Percolation}~\cite{BaloghBM10} -- the \TSS problem with all thresholds fixed to a constant value~$k$.
\end{itemize}

\subsection{Previous Results}
The \TSS problem received attention of researchers in theoretical computer science in the past years.
A general upper bound on the number of selected vertices under majority constraints is $|V|/2$~\cite{AckermanBW10}.
The \TSS problem admits an \FPT algorithm when parameterized by the vertex cover number~\cite{NichterleinNUW13}.
A $t^{\BigO(w)}\poly(n)$ algorithm is known, where $w$ is the tree-width of the input graph and $t$ is an upper-bound on the threshold function~\cite{BenZwiHLN11}, that is, $f(v)\le t$ for every vertex $v$.
This is essentially optimal, as the \TSS problem parameterized by the path-width is \Wh{1} for majority~\cite{ChopinNNW14} and general functions~\cite{BenZwiHLN11}.
The \TSS problem is solvable in linear time on trees~\cite{Chen09} and more generally on block-cactus graphs~\cite{ChiangHLWY13}.
The optimization variant of the \TSS problem is hard to approximate~\cite{Chen09} within a polylogarithmic factor.
For more and less recent results we refer the reader to a survey by Peleg~\cite{Peleg02}.
Cicalese et al.~\cite{Cicalese14,Cicalese15}, considered versions of the problem in which the number of rounds of the activation process is bounded.
For graphs of bounded clique-width, given parameters $a, b, \ell$, they gave polynomial-time algorithms (\XP algorithms) to determine whether there exists a target set of size $b$, such that at least $a$ vertices are activated in at most $\ell$ rounds.
Recently, Hartmann~\cite{Hartmann18} gave a single-exponential \FPT algorithm for \TSS parameterized by clique-width when all thresholds are bounded by a constant.

\subsection{Our Results}
In this work we generalize some results obtained by Nichterlein et al.~\cite{NichterleinNUW13}.
Chopin et al.~\cite{ChopinNNW14} essentially proved that in sparse graph classes (such as graphs with the bounded tree-width) parameterized complexity of the \MAJTSS problem is the same as for the \TSS problem\footnote{There is an \FPT algorithm for both problems parameterized by the vertex cover number, but the problems parameterized by the tree-width are \Wh{1}}.
For these graph classes, it is not hard to see that e.g. if the threshold for vertex $v$ is set above the majority (i.e., $f(v) > \lceil \deg(v)/2 \rceil$), then we may add $2\bigl(f(v) - \lceil \deg(v)/2 \rceil\bigr)$ vertices neighboring $v$ only and the parameter stays unchanged whereas the threshold of $v$ dropped to majority.
However, this is not true in general for dense graph classes.
We demonstrate this phenomenon for the parameterization by the neighborhood diversity.
We show an \FPT algorithm for a function which generalizes both constant and majority threshold functions.
We call this function uniform (and the corresponding problem \UNITSS), see the next section for a proper definition.
Roughly speaking, all vertices belonging to a same part of a graph decomposition must possess the same value of the threshold function.
In slight contrast to previous results, we derive an \FPT algorithm that, instead of the maximal threshold value $t$, depends on the size of the image of the threshold function for graphs having bounded neighborhood diversity.

\begin{theorem}\label{thm:MAJTSSisFPTwrtND}
There is an \FPT algorithm for the \UNITSS problem parameterized by the neighborhood diversity of the input graph.
\end{theorem}

However, the problem is hard if the threshold function is not restricted.
We prove {\sf W[1]}-hardness of the problem and also a lower bound of algorithm running time based on ETH.
The Exponential-Time Hypothesis (ETH) of Impagliazzo and Paturi~\cite{ImpagliazzoP01} asserts that there is no~\( 2^{o(n)} \) algorithm solving the \textsc{Satisfiability} problem, where $n$ is the number of variables.

\begin{theorem}\label{thm:TSSisWwrtND}
The \TSS problem is\/ \Wh{1} parameterized by the neighborhood diversity of the input graph.
Moreover, unless ETH fails, there is no algorithm of running time $f(k)n^{o(k/\log k)}$ solving \TSS, where $k$ is the neighborhood diversity of the input graph.
\end{theorem}

The complexity of the \MAJTSS problem is not resolved for parameterization by the cluster vertex deletion number~\cite{ChopinNNW14} (the number of vertices whose removal from the graph results in a collection of disjoint cliques).
We have a positive result for a slightly stronger parameterization - the twin cover number.
We assume that the edge set between the deletion set and the rest of the graph is somewhat homogeneous, that is, two adjacent vertices in the cluster graph (i.e., a disjoint union of cliques) are true twins.

\begin{theorem}\label{thm:MAJTSSisFPTwrtTC}
There is an \FPT algorithm for the \UNITSS problem parameterized by the twin cover number of the input graph.
\end{theorem}

\begin{theorem}\label{thm:TSSisWwrtTC}
The \TSS problem is\/ \Wh{1} parameterized by the twin cover number of the input graph.
Moreover, unless ETH fails, there is no algorithm of running time $f(k)n^{o(k/\log k)}$ solving \TSS, where $k$ is the twin cover number of the input graph.
\end{theorem}

Previous result~\cite{ChopinNNW14} implies that the parameterized complexity of the \TSS and the \MAJTSS problems is the~same in~graphs with bounded clique-width.
Of course, much more is known -- the proof therein shows that the \MAJTSS problem is \Wh{1} on graphs of the bounded tree-depth (even though only the tree-width is claimed).
We show that this is also the case for parameterization by the (restricted) modular-width which is a parameter generalizing both the neighborhood diversity and the twin cover number.

\begin{theorem}\label{thm:MAJTSSisWwrtMW}
The \MAJTSS problem is\/ \Wh{1} parameterized by the modular-width of the input graph.
Moreover, unless ETH fails, there is no algorithm of running time $f(k)n^{o(k/\log k)}$ solving \MAJTSS, where $k$ is the modular-width of the input graph.
\end{theorem}

\begin{figure}[bt]
\centering
  \newcommand{\distX}{1.9cm}
\newcommand{\distY}{.65cm}

\usetikzlibrary{positioning, calc, fit}

\begin{tikzpicture}[node  distance=\distY, font=\small]
  \tikzstyle{result}=[minimum width=\distX, fill, minimum height=\distY]
  \tikzstyle{FPTResult}=[result, color=green!40, text=black, draw=black]
  \tikzstyle{WhardResult}=[result, color=red!80, text=black, draw=black]
  \tikzstyle{pNPhardResult}=[result, color=black, text=white]
  \tikzstyle{unknown}=[result, color=white, text=black, draw=black]
  \tikzstyle{parameter}=[text width=\distX + .3cm, text centered]
  \tikzstyle{edge}=[ultra thick]
  \tikzstyle{missingEdge}=[ultra thick, dashed]
%  \node[FPTResult] (XX) {Th 1};
%  \node[WhardResult, right of=XX] (XX) {Th 3};
%  \node[pNPhardResult, right of=XX] (XX) {Th 3};

  \newcommand{\trojice}[9]{
    \node[parameter] at #9 (ParLabel) {#1};
    \node[#2] at ($(ParLabel.north)!0.5!(ParLabel.south) - (0,0.8cm)$) (ResConst) {#3};
    \node[#4, below of=ResConst] (ResMAJ) {#5};
    \node[#6, below of=ResMAJ] (ResGeneral) {#7};

    \node[fit=(ParLabel)(ResConst)(ResMAJ)(ResGeneral), draw] (#8) {};
  }

  \trojice{Vertex Cover}{FPTResult}{\cite{ChopinNNW14}}{FPTResult}{\cite{ChopinNNW14}}{FPTResult}{\cite{ChopinNNW14}}{VC}{(0,0)}

  \trojice{Neighborhood Diversity}{FPTResult}{Thm \ref{thm:MAJTSSisFPTwrtND}}{FPTResult}{Thm \ref{thm:MAJTSSisFPTwrtND}}{WhardResult}{Thm~\ref{thm:TSSisWwrtND}}{ND}{(3.5,-1.7)}

  \trojice{Twin Cover}{FPTResult}{Thm \ref{thm:MAJTSSisFPTwrtTC}}{FPTResult}{Thm \ref{thm:MAJTSSisFPTwrtTC}}{WhardResult}{Thm \ref{thm:TSSisWwrtTC}}{TC}{(3.5,1.7)}

  \trojice{Modular Width}{FPTResult}{\cite{Hartmann18}}{WhardResult}{Thm~\ref{thm:MAJTSSisWwrtMW}}{WhardResult}{Thm~\ref{thm:MAJTSSisWwrtMW}}{MW}{(7,0)}

 % \trojice{Shrub Depth}{FPTResult}{\cite{Hartmann18}}{WhardResult}{\cite{ChopinNNW14}, Thm~\ref{thm:MAJTSSisPNPwrtSD}}{WhardResult}{\cite{ChopinNNW14}, Thm~\ref{thm:MAJTSSisPNPwrtSD}}{SD}{(10.5,2)}

  \draw[edge] (VC) -- (ND);
  \draw[edge] (VC) -- (TC);
  \draw[edge] (TC) -- (MW);
  \draw[edge] (ND) -- (MW);

  %\draw[missingEdge] (MW) -- (SD);
\end{tikzpicture}
  \caption{A map of considered parameterizations and specializations of the \TSS problem. The three boxes represent from the top the constant threshold \TSS, \MAJTSS, and \TSS problems. Green (light gray) boxes represent \FPT algorithms, red (dark gray) boxes represent \W{1}-hardness result.}
\end{figure}

\section{Preliminaries}\label{sec:StructuralGraphParameters}
In this section we give formal definitions of several graph parameters used in this work.
To get better acquainted with these parameters and their relations, we provide a Hasse diagram of the considered parameters in Figure~\ref{fig:parameterMap}.

%\begin{definition}[Vertex cover]\label{def:vc}
For a graph $G = (V,E)$ the set $U\subseteq V$ is called a \emph{vertex cover} of $G$ if for every edge $e\in E$ it holds that $e\cap U\neq\emptyset.$
The \emph{vertex cover number} of a graph, denoted as $\vc(G)$, is the least integer $k$ for which there exists a vertex cover of size $k$.

\begin{figure}[bt]
  \begin{minipage}{0.28\textwidth}
    \usetikzlibrary{arrows,positioning,backgrounds,fit,shapes,calc,scopes}

\begin{tikzpicture}[]
  % arrow styles
  \tikzstyle{arrowBasic} = [->, >=stealth, very thick]
  \tikzstyle{linDep} = [arrowBasic]
  \tikzstyle{expDep} = [arrowBasic, gray!70]

  % definition of nodes
  \node (cw) {$\cw$};

  %\node[below left of=cw] (sd) {sd};
  \node[below of=cw] (mw) {$\mw$};
  \node[right of=mw] (tw) {$\tw$};

  \node[below of=mw] (rmw) {rmw};
  \node[left of=rmw] (cvd) {$\cvdn$};

  \node[below of=rmw] (nd) {$\nd$};
  \node[left of=nd] (tc) {$\tc$};
  \node[right of=nd] (pw) {$\pw$};

  \node[below of=nd] (vc) {$\vc$};

  % definition of arrwos
  \draw[expDep] (vc) -- (nd);
  \draw[linDep] (vc) -- (tc);
  \draw[linDep] (vc) -- (pw);

  \draw[linDep] (pw) -- (tw);
  \draw[expDep] (tc) -- (rmw);
  \draw[linDep] (nd) -- (rmw);
  \draw[linDep] (tc) -- (cvd);
  %\draw[expDep] (cvd) -- (sd);
  \draw[linDep] (rmw) -- (mw);
  \draw[expDep] (tw) -- (cw);
  %\draw[linDep] (sd) -- (cw);
  %\draw[expDep] (rmw) -- (sd);
  \draw[expDep] (cvd) -- (cw);
  \draw[linDep] (mw) -- (cw);

  % backgrounds
  \begin{scope}[on background layer]
    \node[fill=yellow!20, fit=(vc)(pw)(tw), rounded corners] {};
  \end{scope}

  \begin{scope}[on background layer]
    \node[fill=orange!50, fit=(nd)(cvd)(tc)(cw), rounded corners] {};
  \end{scope}
\end{tikzpicture}
  \end{minipage}\hfill
  \begin{minipage}{0.7\textwidth}
    \caption{A map of the considered parameters.
	A black arrow stands for a linear upper bound, while a gray arrow stands for an exponential upper bound.
	That is, if a graph $G$ has $\vc(G)\le k$, then \mbox{$\nd(G)\le 2^k + k$}.
   } \label{fig:parameterMap}
  \end{minipage}
\end{figure}

As the vertex cover number is (usually) too restrictive, many authors focused on defining other (i.e., weaker) structural parameters.
Three most well-known parameters of this kind are the path-width, the tree-width (introduced by Robertson and Seymour~\cite{RS86:GMII}), and the clique-width (introduced by Courcelle et al.~\cite{CMR98}).
Classes of graphs with the bounded tree-width (respectively the path-width) are contained in the so-called sparse graph classes.

There are (more recently introduced) structural graph parameters which also generalize the vertex cover number but, in contrast with {e.g.}\ the tree-width, these parameters focus on dense graphs.
First, up to our knowledge, of these parameters is the \emph{neighborhood diversity} defined by Lampis~\cite{Lampis12}.

\subparagraph{Neighborhood Diversity}
We say that two distinct vertices $u,v$ are of the same \emph{neighborhood type} if they share their respective neighborhoods, that is, when ${N(u)\setminus\{v\} = N(v)\setminus\{u\}}.$

\begin{definition}[Neighborhood diversity~\cite{Lampis12}]\label{def:nd}
A graph $G = (V,E)$ has the \emph{neighborhood diversity} at most $w$ ($\nd(G)\le w$) if there exists a decomposition $\mathcal{D}_{\nd} = (C_i)^w_{i=1}$ of\/ $V = C_1 \overset{\cdot}{\cup} \cdots \overset{\cdot}{\cup} C_w$ (we call the sets $C_i$ \emph{types}) such that all vertices in a type have the same neighborhood type (i.e., are twins).
\end{definition}

Note that every type induces either a clique or an independent set in $G$ and two types are either joined by a complete bipartite graph or no edge between vertices of the two types is present in $G.$
Thus, we use a notion of a \emph{type graph}, that is, a graph~$T_G$ representing the graph $G$ and its neighborhood diversity decomposition in the following way.
The vertices of the type graph $T_G$ are the types $C_1,\dots,C_w$ and two such vertices are joined by an edge if all the vertices of corresponding types are adjacent.
It is possible to compute the neighborhood diversity of a graph in linear time~\cite{Lampis12}.

\subparagraph{Twin Cover}
If two vertices $u,v$ have the same neighborhood type and $e = \{u,v\}$ is an edge of the graph, we say that $e$ is a \emph{twin edge}.

\begin{definition}[Twin cover number~\cite{Ganian11}]
\label{def:tc}
A set of vertices $T\subseteq V$ is a \emph{twin cover} of a graph $G=(V,E)$ if for every edge $e\in E$ either $T \cap e \neq \emptyset$ or $e$ is a twin edge.
We say that $G$ has the \emph{twin cover number} $t$ (${\tc(G) = t}$) if the size of a minimum twin cover of $G$ is $t.$
\end{definition}

Note that after removing $T$ from the graph $G$, the resulting graph is a disjoint union of cliques -- we call them \emph{twin cliques}.
Moreover, for every vertex $v$ in $T$ and a twin clique $C$ it holds that $v$ is either adjacent to every vertex in $C$ or to none of them.
A \emph{twin cover decomposition} $\mathcal{D}_{\tc} = (C_i)^\nu_{i = 1}$ of a graph $G$ is a partition of $V(G)$ such that each $C_i$ is either a vertex of the twin cover or a twin clique.

Note that the twin cover number can be upper-bounded by the vertex cover number.
The structure of graphs with bounded twin cover is very similar to the structure of graphs with bounded vertex cover number.
Thus, there is a hope that many of known algorithms for graphs with bounded vertex cover number can be easily turned into algorithms for graphs with bounded twin cover number~\cite{Ganian11}.

\subparagraph{Uniform Threshold Function}
As it is possible to compute the neighborhood diversity (or the twin cover) decomposition in polynomial time (or \FPT-time, respectively), we may assume that the decomposition is given in the input.
Given a decomposition $\mathcal{D}$ ($\mathcal{D}_{\nd}$ or $\mathcal{D}_{\tc}$) a threshold function $f\colon V(G) \to \N$ is \emph{uniform with respect to $\mathcal{D}$} if $f(u) = f(v)$ for every $u,v\in C$ and every $C\in\mathcal{D}$.
Observe that this notion generalizes the previously studied model~\cite{ChopinNNW14} in which the threshold function is required to satisfy $f(u) = f(v)$ whenever $|N(u)| = |N(v)|$, since this indeed holds if $u,v \in C$ and $C \in \mathcal{D}$.
It is not hard to see that the uniform function generalizes both the constant and the majority functions for the twin cover number and the neighborhood diversity.

Moreover, if $f(v)$ is bounded by a constant $c$ for all $v\in V(G)$, then there exists $\mathcal{D_{\nd}}$ with $| \mathcal{D}_{\nd} | \le c\cdot\nd(G)$ such that $f$ is uniform with respect to $\mathcal{D_{\nd}}$.
We stress here that this construction is not legal for the twin cover decompositions.
\UNITSS is a variant of \TSS, where the input instance $\left( G, f, b, \mathcal{D} \right)$ is restricted in such a way that the function $f$ is uniform with respect to $\mathcal{D}$.

\subparagraph{Modular-width}
Both the neighborhood diversity and the twin cover number are generalized by the modular-width.
Here we deal with graphs created by an algebraic expression that uses the following four operations:
\begin{enumerate}
  \item Create an isolated vertex.
  \item The \emph{disjoint union} of two graphs, that is from graphs ${G = (V,E)}, {H = (W, F)}$ create a graph ${(V\cup W, E\cup F)}$.
  \item The \emph{complete join} of two graphs,
that is from graphs $G = (V,E), H = (W, F)$ create a graph with vertex set $V\cup W$ and edge set ${E\cup F\cup \Setof{\{v,w\}}{v\in V, w\in W}}$. Note that the edge set of the resulting graph can be also written as ${E\cup F\cup (V\times W)}.$
  \item The \emph{substitution operation} with respect to a template graph $T$ with vertex set ${\{v_1,v_2,\dots,v_k\}}$ and graphs ${G_1,G_2,\dots,G_k}$ created by an algebraic expression; here $G_i = (V_i, E_i)$ for $i =1,2,\ldots,k$. The substitution operation, denoted by ${T(G_1,G_2, \dots, G_k)},$ results in the graph on vertex set ${V = V_1\cup V_2\cup\dots\cup V_k}$ and edge set
\[
  E = E_1\cup E_2\cup\dots\cup E_k\cup \bigcup_{\{v_i,v_j\}\in E(T)} \Setof{\{u,v\}}{ u\in V_i, v\in V_j} \,.
\]
\end{enumerate}

\begin{definition}[Modular-width~\cite{GLO13}]
\label{def:mw}
Let $A$ be an algebraic expression that uses only operations {\rm 1--4} above.
The \emph{width} of the expression $A$ is the maximum number of operands used by any occurrence of operation $4$ in $A$.
The \emph{modular-width} of a graph $G$, denoted $\mw(G)$, is the least positive integer $k$ such that $G$ can be obtained from such an algebraic expression of width at most $k.$
\end{definition}
An algebraic expression of width $\mw(G)$ can be computed in linear time~\cite{TCHP08}.

\subparagraph{Restricted Modular-width}
We would like to introduce here a restriction of the modular-width that still generalizes both the neighborhood diversity and the twin cover number.
The algebraic expression used to define a graph $G$ may contain the substitution operation at most once and if it contains the substitution operation it has to be the last operation in the expression.
However, there is no limitation for the use of operations {\rm 1--3}.

\section{Positive Results}
In this section we give proofs of Theorem~\ref{thm:MAJTSSisFPTwrtND} and~\ref{thm:MAJTSSisFPTwrtTC}.
In the first part we discuss the crucial property of dense structural parameters we study here -- the uniformity of neighborhoods.
This, opposed to e.g. the cluster vertex deletion number, allows us to design parameterized algorithms.
In this section by a decomposition $\mathcal{D}$ we mean a neighborhood diversity decomposition or a twin cover decomposition.

\begin{lemma}\label{lem:MAJtrueActivation}
Let $G = (V,E)$ be a graph, $\mathcal{D}$ be a decomposition of $G$, $S\subseteq V$ be a target set, and $f$ be a uniform threshold function with respect to $\mathcal{D}$.
Let $C \in \mathcal{D}$ and $S_0 = S, S_1, \ldots$ be an activation process arising from $S$.
There exists $j_C \in \N_0$ such that:
\begin{enumerate}
  \item For all $i < j_C$ holds that $S_i \cap C = S_0 \cap C$.
  \item For all $i \geq j_C$ holds that $S_i \cap C = C$.
\end{enumerate}
%Moreover, there exists $j$ with $j \in \N_0$ such that for $C$ the first item applies in rounds $0,\ldots, j$ and the second in rounds $j+1,\ldots$.
\end{lemma}
\begin{proof}
For $i=0$ clearly holds Part 1.
Moreover, if $S_0 \cap C = C$, then $j_C = 0$ and the lemma holds.

Now suppose there is $u \in V(C) \setminus S_0$ which is activated in the round $j$.
All vertices in $V(C) \setminus S_0$ have the same threshold as $v$ and the same neighborhood.
Thus, all vertices in $V(C) \setminus S_0$ is activated in the round $j$ as well and we can define $j_C$ as $j$.
\end{proof}

Let $C \in \mathcal{D}$.
For a threshold function $f$ which is constant on $C$ we define $f'(C)$ as $f(v)$ for arbitrary vertex $v$ in $C$.
We say that $C$ is \emph{activated in round $j_C$}, given by Lemma~\ref{lem:MAJtrueActivation}.
We denote $a^S_i(v)$ the number $|S_{i-1} \cap N(v)|$, i.e., the number of active neighbors of $v$ in the round $i-1$ in the activation process arising from the set $S$.
Thus, a vertex $v$ is activated in the first round $i$ when $a^S_i(v) \geq f(v)$ holds but $a^S_{i-1}(v) < f(v)$.
We can suppose that for every vertex $v \in V(G)$ holds that $f(v) \leq \deg(v)$, otherwise there is clearly no target set for $G$.

\subsection{Uniformity and Twin Cover}
In this subsection we present an algorithm for \UNITSS parameterized by the twin cover number.

\subparagraph{Trivial Bounds on the Minimum Target Set}
Let $G = (V,E)$ be a graph with a twin cover $T$ of size $t$ and let $C_1, C_2, \ldots, C_q$ be the twin cliques of $G$.
For a twin clique $C$ by $N(C)$ we denote the common twin cover neighborhood, that is, $N(v)\cap T$ for any $v\in C$.
We show there is a small number of possibilities how the optimal target set can look like.
Let $b_{C} = \max\bigl(f'(C) - |N(C)|, 0\bigr)$ for a twin clique $C$.
Observe that we need to select at least $b_C$ vertices of the twin clique $C$ to any target set.

The first preprocessing we do here is with twin cliques~$C$ having~$N(C) = \emptyset$.
Clearly, since \TSS is solvable in polynomial time on a clique~$C$, we can remove~$C$ from~$G$ and reduce the budget~$b$ accordingly.
Thus, from now we assume $N(C) \neq \emptyset$ for all twin cliques~$C$.

\begin{observation}
\label{obs:OptimalSetTC}
If the minimum target set of $G$ has size $s$, then $B \le s \le B  + t$ for $B = \sum_{i = 1}^q b_{C_i}$.
\end{observation}
\begin{proof}
Let $S$ be a target set for $G$ of size $s$.
Suppose there is a twin clique $C$ such that $\bigl|S \cap V(C)\bigr| = p < b_C$.
It means that $b_C > 0$.
Let $v \in C \setminus S$.
Note that $p < |V(C)|$ as $f'(C) \leq |C| + |N(C)|$.
Thus, the vertex $v$ exists.
For the vertex $v$ it holds that $a^S_i(v) \le p + |N(C)|$ for every round $i$ of the activation process.
Thus, the vertex $v$ is never activated because $p + |N(C)| < b_C + |N(C)| = f'(C)$ and $S$ is not a target set.
On the other hand, if we put $b_C$ vertices from each twin clique $C$ into a set $S'$, then the set $S' \cup T$ is a target set because every vertex not in $S'$ is activated in the first round.
\end{proof}

\subparagraph{Structure of the Solution}
Let $(G,f,b,\mathcal{D}_{\tc})$ be an instance of \UNITSS with $\tc(G) = t$.
By Observation~\ref{obs:OptimalSetTC}, if $b < \sum b_C$, then we automatically reject.
On the other hand, if $b \geq t + \sum b_C$, then we automatically accept.
Let $w = b - \sum b_C$.
Thus, to find a target set of size $b$ we need to select $w$ \emph{excess vertices} from the twin cliques and the twin cover.
We will show there are at most $g(t)$ interesting choices for these $w$ excess vertices for some computable function $g$ and those choices can be efficiently generated.
Since we can check if a given set $S \subseteq V(G)$ is a target set in polynomial time, we conclude there is an {\FPT} algorithm for \UNITSS.

We start with an easy preprocessing.
Let $C$ be a twin clique with $b_C > 0$.
We select $b_C$ vertices $V' \subseteq C$ and remove them from the graph $G$.
We also decrease the threshold value by $b_C$ of every vertex which was adjacent to $V'$ (recall that vertices in $V'$ have the same neighborhood type, thus any vertex adjacent to some vertex in $V'$ is adjacent to all vertices in $V'$).
Formally, we get an equivalent instance $(G_1,f_1, b - b_C, \mathcal{D}'_{\tc})$, where $G_1$ is $G$ without vertices $V'$, $\mathcal{D}'_{\tc}$ is $\mathcal{D}_{\tc}$ restricted to $V(G_1)$ and
\[
 f_1(v) = \begin{cases}
       f(v) & v \not\in N_G(V') \\
       f(v) - b_C & v \in N_G(V').
       \end{cases}
\]
It is easy to see that the instances $(G,f,b,\mathcal{D}_{\tc})$ and $(G_1,f_1,b-b_C,\mathcal{D}'_{\tc})$ are equivalent, since any target set for $G$ needs at least $b_C$ vertices in the twin clique $C$ due to Observation~\ref{obs:OptimalSetTC}.
Note that the function $f_1$ is uniform with respect to $\mathcal{D}'_{\tc}$; which follows from the same fact as for the function~$f$.
We repeat this process for all twin cliques.
Furthermore, if there is a vertex~$v$ with~$f_1(v) \le 0$, we delete it and decrease~$f_1(w)$ by~1 for each~$w \in N(v)$.
From now on, we suppose that the instance $(G,f',b,\mathcal{D}_{\tc})$ is already preprocessed.
Thus, for every twin clique $C$ it holds that $b_C = 0$ and $f'(C) \leq |N(C)| \leq t$.

We say that a twin clique $C$ is of a \emph{type} $(Q,r)$ for $Q\subseteq T$ and $1 \le r \leq t$ if $Q = N(C)$ and $f'(C) = r$.
%Two twin cliques $C$ and $D$ are of the same type if $N(C) = N(D)$ and $f'(C) = f'(D)$.
Note that there are at most $t\cdot 2^t$ distinct types of the twin cliques.

We start to create a possible target set $S$ of size $b$.
We add $w_1$ (for some $w_1 \leq w$) vertices from the twin cover $T$ to $S$ (there are at most $2^t$ such choices).
Now we need to select $w_2 = w - w_1$ excess vertices from twin cliques to $S$.

The number of the twin cliques of one type may be large.
Thus, for the twin cliques we need some more clever way than try all possibilities.
The intuition is that if we want to select some excess vertices from a clique of type $(Q,r)$ it is a ``better'' choice to select the vertices from large cliques of type $(Q,r)$.
We assign to each type $(Q,r)$ a number $w_{(Q,r)}$ how many excess vertices would be in twin cliques of type $(Q,r)$.
We prove that it suffices to distribute $w_{(Q,r)}$ excess vertices among the $w_{(Q,r)}$ largest twin cliques of the type $(Q,r)$.

\begin{definition}
\label{def:leaky}
 Let $C_1,\dots,C_p$ be twin cliques of type $(Q,r)$ ordered by the size in a descending order, i.e., for all $1 \le i < p$ it holds that $|C_i| \ge |C_{i+1}|$.
 We say that a target set has a \emph{hole} $(C_i, C_j)$ for $j > i$ if $|S \cap C_i| = 0$ and $|S \cap C_{j}| \ge 1$.
 A target set is \emph{$(Q,r)$-leaky} if it has a hole and it is \emph{$(Q,r)$-compact} otherwise.
\end{definition}
Our goal is to prove that if there is a target set $S$ which is $(Q,r)$-leaky, then there is also a target set $R$ which is $(Q,r)$-compact and $|R| = |S|$.
After that we prove that there is only small number of compact targets, thus we can try all of them.

\begin{lemma}
\label{lem:leaky}
 Suppose $S$ is a target set for a graph $G$ with a threshold function~$f$ and $S$ is $(Q,r)$-leaky for some twin clique type $(Q,r)$.
 Then, there is a target set $\tilde{S}$ such that:
 \begin{enumerate}
  \item It holds that $|\tilde{S}| = |S|$.
  \item The sets $\tilde{S}$ and $S$ differ only at the twin cliques of the type $(Q,r)$.
  \item The set $\tilde{S}$ is $(Q,r)$-compact.
 \end{enumerate}
\end{lemma}

As the proof is a bit technical, we postpone it after we finish the algorithm using the above lemma.

\subparagraph{Finishing the Algorithm}
By Lemma~\ref{lem:leaky} we know that if there is a $(Q,r)$-leaky target set $S$, then there is a $(Q,r)$-compact target set $\tilde{S}$ of the same size.
Moreover, sets $S$ and $\tilde{S}$ differ only at the twin cliques of the type $(Q,r)$.
Thus, if we repeat the procedure for all types, we get a target set without any hole.
To summarize how to distribute $w$ excess vertices:
\begin{enumerate}
\item Pick $w_1$ vertices from the twin cover $T$, in total $2^t$ choices.
\item Distribute $w_2 = w - w_1$ excess vertices among $t\cdot2^t$ types of the twin cliques, in total $(t\cdot2^t)^t = 2^{\BigO(t^2)}$ choices.
\item Distribute $w_{(Q,r)}$ excess vertices among the $w_{(Q,r)}$ largest cliques of type $(Q, r)$, in total $t^t$ choices.
\end{enumerate}
By this we create $2^{\BigO(t^2)}$ candidates for a target set.
For each candidate we test whether it is a target set for $G$ or not.
If any candidate set is a target set, then we find a target set of size $b$.
If no candidate set is a target set, then by argumentation above we know the graph $G$ has no target set of size (at most) $b$.
This finishes the proof of Theorem~\ref{thm:MAJTSSisFPTwrtTC}.

\subsubsection*{Proof of Lemma~\ref{lem:leaky}}

\subparagraph{Delayed Activation Process}
For a better analyzing the activation process we introduce a little bit more general notion.
Besides the threshold function we also have a delay function $d \colon V(G) \to \N$.
Let $S \subseteq V(G)$.
The \emph{delayed activation process} arising from the set~$S$ is $S = S_0, S_1, \dots$, where
\[
  S_{i + 1} = S_i \cup \setof{v \in V}{|N(v) \cap S_i| \geq f(v) \text{ and } i + 1 \geq d(v)}.
\]
Thus, a vertex $v$ cannot be activated before the round $d(v)$.
We say the set $S$ is a \emph{delayed target set} if there is some $j$ such that $S_j = V(G)$.
Note that a target set is a special case of a delayed target set for a delay function $d(v) = 0$ for all vertices $v \in V(G)$.

\begin{observation}
\label{obs:delay}
 Let $S$ be a delayed target set for a graph $G$, a threshold function~$f$, and a delay function~$d$.
 Then, $S$ is a target set for the graph $G$ and the threshold function $f$.
\end{observation}
\begin{proof}
 The delayed condition of vertex activation is more restrictive than the condition for the standard activation process.
 Suppose a vertex $v$ is activated in the round~$i$ during a delayed activation process.
 Then, the vertex $v$ is certainly activated during the round $i$ of the standard activation process at the latest.
\end{proof}

Now we are ready to prove Lemma~\ref{lem:leaky}.
The idea of the proof is to create the target set $\tilde{S}$ from $S$ by removing vertices from $C_j$ and adding the same number of vertices from $C_i$ for any hole $(C_i, C_j)$.
\begin{proof}[Proof of Lemma~\ref{lem:leaky}]
 Let $T$ be a twin cover of the graph~$G$.
 Let $C_1,\dots,C_p$ be order of all twin cliques of the type $(Q,r)$ as in Definition~\ref{def:leaky} and let $(C_i,C_j)$ be a hole in the target set $S$.
 Let ${\cal S} = (S = S_0, S_1,\dots)$ be an activation process arising from $S$ and the clique $C_j$ is activated in the round $j'$ and the clique $C_i$ in the round $i'$ of the activation process ${\cal S}$.

 First, we show that $j' \leq i'$.
 Since $C_i$ is activated in the round $i'$ and $S_{i'-1} \cap C_i = \emptyset$, there is $f'(C_i)$ active vertices in $Q \subseteq T$ in the round $i' - 1$.
 Since $f'(C_j)=f'(C_i) \leq |S_{i'} \cap Q|$, the clique $C_j$ is activated at most in the round $i'$ (i.e., $C_j \subseteq S_{i'}$).
 Therefore, $j' \leq i'$.

 \begin{figure}[bt]
   \begin{minipage}{0.42\textwidth}
     \begin{tikzpicture}
  \usetikzlibrary{calc,arrows}
  \newcommand{\topheight}{1cm}
  \newcommand{\midheight}{1.5cm}
  \newcommand{\botheight}{1cm}
  \newcommand{\arrowshift}{.4cm}

  \tikzstyle{toppart}=[minimum height=\topheight, minimum width=1cm, draw, inner sep=0pt, fill=black!20]
  \tikzstyle{midpart}=[minimum height=\midheight, minimum width=1cm, draw, inner sep=0pt]
  \tikzstyle{botpart}=[minimum height=\botheight, minimum width=1cm, draw, inner sep=0pt, fill=green!60]
  \tikzstyle{pomocne}=[dotted, thick]

\node[toppart, label={180:$0$}, label={0:$i'$}, label={90:$C_i$}] (Ltop) {$X$};
\node[midpart, label={180:$j'$}, label={0:$i'$}] at ($(Ltop.south) - (0, \midheight / 2)$) (Lmid) {$Y$};
\node[botpart, label={180:$i'$}, label={0:$i'$}] at ($(Lmid.south) - (0, \botheight / 2) + (0, .1pt)$) (Lbot) {};

\node[toppart, label={180:$i'$}, label={0:$0$}, label={90:$C_j$}] (Rtop) at ($(Ltop) + (3cm, 0)$) {$S$};
\node[midpart, label={180:$i'$}, label={0:$j'$}] (Rmid) at ($(Rtop.south) - (0, \midheight / 2)$) (Rmid) {};

\draw[<-, ultra thick,>=stealth] ($(Ltop.west)!.7!(Ltop.east) + (0,\arrowshift)$) to[out=45] ($(Rtop.west)!.2!(Rtop.east)+ (0,\arrowshift)$);

\draw[pomocne] (Ltop.north east) -- (Rtop.north west);
\draw[pomocne] (Lmid.north east) -- (Rmid.north west);
\draw[pomocne] (Lbot.north east) -- (Rmid.south west);
\end{tikzpicture}
   \end{minipage}\hfill
   \begin{minipage}{0.56\textwidth}
     \caption{An illustration of how we create the set $R$ from the set $S$.
     There are rounds when the parts of cliques are activated during the activation process ${\cal S}$ or ${\cal R}$ on the right or left, respectively.}
     \label{fig:TCTwoCliques}
   \end{minipage}
 \end{figure}

 We create a delayed target set $R$ from~$S$ by removing vertices from $C_j$ and adding the same number of vertices from $C_i$; see Figure~\ref{fig:TCTwoCliques}.
 Formally,
 \[
   R = \bigl(S \setminus C_j\bigr) \cup X \,,
 \]
where $X \subseteq C_i$ and $|X| = |S \cap C_j|$.
Let $Y \subseteq C_i \setminus X$ such that $|Y| = |C_j \setminus S|$.
We set the delay function $d$ as follows
\[
 d(v) =
 \begin{cases}
    j' & v \in Y, \\
    i' & v \in \bigl(C_i \setminus (X \cup Y)\bigr) \cup C_j, \\
    0 & \text{otherwise}.
 \end{cases}
\]
Let ${\cal R} = (R = R_0,R_1,\dots)$ be an activation process arising from the set $R$.
The meaning of the delay~$d$ is that parts of the clique $C_i$ (the sets $X$ and $Y$) will behave in the activation process ${\cal R}$ in the same manner as the clique $C_j$ in the activation process~$\mathcal{S}$.
%We will prove the following claim.
\begin{claimproof}
\label{clm:NewTS}
~\begin{enumerate}
 \item For every $k$ it holds that $S_k \cap T \subseteq R_k \cap T$ and $S_k \cap C \subseteq R_k \cap C$ for every twin clique $C$ different from $C_i$ and $C_j$.
 \item For every vertex $v$ in the cliques $C_i$ and $C_j$ holds that $v \in R_{d(v)}$.
\end{enumerate}
\end{claimproof}
Immediate corollary of Claim~\ref{clm:NewTS} is that $R$ is a delayed target set such that $|R| = |S|$.
By Observation~\ref{obs:delay}, we get that $R$ is also a target set.
The set $R$ contains at least one less hole than the set $S$.
Therefore, if we repeat this procedure we eventually get the sought $(Q,r)$-compact target set $\tilde{S}$.
\end{proof}

\begin{proof}[Proof of Claim~\ref{clm:NewTS}]
We prove Claim~\ref{clm:NewTS} by induction on the number of round $k$.
It is clear it holds for $k = 0$ because of the construction of the set $R_0 = R$.
Now suppose it holds for $k \geq 0$ and we will prove it for $k + 1$.
Recall that for every vertex $v$ not in the cliques $C_i$ and $C_j$ holds that $d(v) = 0$.
Let $v \in (S_{k + 1} \setminus S_k)  \cap T$.
Let ${\cal C}_v$ be the set of twin cliques $C$ such that $v \in N(C)$ and $T_v = T \cap N(v)$, i.e., the neighbors of $v$ in $T$.
Then,
\[
 a^S_{k+1}(v) = \sum_{C \in {\cal C}_v} |S_{k} \cap C| + |S_k \cap T_v| \geq f(v).
\]
Suppose $v \not \in Q$.
Then, by induction hypothesis (Part 1 in the claim), it holds that
\begin{align*}
  f(v)
  \leq
  a^S_{k+1}(v)
  =
  \sum_{C \in {\cal C}_v} |S_{k} \cap C| + |S_k \cap T_v|
  \leq
  \sum_{C \in {\cal C}_v} |R_{k} \cap C| + |R_k \cap T_v|
  =
  a^R_{k+1}(v) \,.
\end{align*}
Therefore, $v \in R_{k + 1} \cap T$.

Now suppose $v \in Q$.
We distinguish three cases: $k + 1 \leq j'$, $j' < k + 1 \leq i'$, and $i' < k + 1$.
Suppose $k + 1 \leq j'$.
Thus, the vertex $v$ is activated at most in the same round as the clique $C_j$ in the activation process ${\cal S}$.
Note that in the activation process ${\cal R}$ the only active neighbors of the vertex $v$ in the clique $C_i$ and $C_j$ are those vertices in the set $X$ (which has the same size as $S_0 \cap C_j$).
Then,
\begin{align*}
 f(v)
 \leq
 a^S_{k+1}(v)
 &=
 \sum_{C \in {\cal C}_v} |S_{k} \cap C| + |S_k \cap T_v| \\
 &=
 \sum_{C \in {\cal C}_v \setminus \{C_i, C_j\}} |S_{k} \cap C| + |S_0 \cap C_j| + |S_k \cap T_v| \\
 &\leq
 \sum_{C \in {\cal C}_v \setminus \{C_i, C_j\}} |R_{k} \cap C| + |X| + |R_k \cap T_v|
 =
 a^R_{k+1}(v) \,.
\end{align*}
Thus, $v \in R_{k + 1}$.
The other cases are similar, now suppose that $j' < k + 1 \leq i'$, i.e., the vertex $v$ is activated in the activation process ${\cal S}$ after the clique $C_j$ but before the clique $C_i$.
We consider this case only if $j' < i'$.
In the activation process ${\cal R}$ the vertices in $X$ and $Y$ are the only active vertices in the cliques $C_i$ and $C_j$ (recall that $|X| + |Y| = |C_j|$).
\begin{align*}
 f(v)
 \leq
 a^S_{k+1}(v)
 &=
 \sum_{C \in {\cal C}_v} |S_{k} \cap C| + |S_k \cap T_v|
 =
 \sum_{C \in {\cal C}_v \setminus \{C_i, C_j\}} |S_{k} \cap C| + |C_j| + |S_k \cap T_v| \\
 &\leq
 \sum_{C \in {\cal C}_v \setminus \{C_i, C_j\}} |R_{k} \cap C| + |X| + |Y| + |R_k \cap T_v|
 =
 a^R_{k+1}(v)
\end{align*}
The last case $i' < k + 1$ is when the vertex $v$ is activated after the activation of both cliques $C_i$ and $C_j$ in the both processes.
\begin{align*}
 f(v)
 \leq a^S_{k+1}(v) &= \sum_{C \in {\cal C}_v} |S_{k} \cap C| + |S_k \cap T_v| \\
 &= \sum_{C \in {\cal C}_v \setminus \{C_i, C_j\}} |S_{k} \cap C| + |C_j| + |C_i| + |S_k \cap T_v| \\
 &\leq \sum_{C \in {\cal C}_v \setminus \{C_i, C_j\}} |R_{k} \cap C| + |C_j| + |C_i| + |R_k \cap T_v|
 = a^R_{k+1}(v).
\end{align*}
In all three cases the vertex $v$ has more than $f(v)$ active neighbors in the round $k + 1$ of the process ${\cal R}$.
Therefore, $v \in R_{k + 1}$.

Now let $v \in (S_{k + 1} \setminus S_k) \cap C$ for some twin clique $C \neq C_i,C_j$.
Let $(Q',r')$ be a type of $C$.
Thus by induction hypothesis,
\[
 f(v) \leq a^S_{k+1}(v) = |S_0 \cap C| + |S_k \cap Q'| \leq |R_0 \cap C| + |R_k \cap Q'| = a^R_{k+1}(v).
\]
Therefore, $v \in R_{k + 1}$.

It remains to prove the claim about the cliques $C_i$ and $C_j$.
It is clear that for a vertex $v \in X$ holds that $v \in R_0$.

Let $v \in C_j$, thus $d(v) = i'$.
In the activation process ${\cal S}$ the clique $C_i$ is activated in the round $i'$, thus $f'(C_i) \leq |S_{i' - 1} \cap Q|$.
The cliques $C_i$ and $C_j$ have the same type, thus $f(v) = f'(C_i)$.
Therefore by induction hypothesis,
\[
  f(v) \leq |S_{i' - 1} \cap Q| \leq |R_{i' - 1} \cap Q| = a_{i'}^R(v)
\]
and the clique $C_j$ is activated in the round $i'$ of the activation process ${\cal R}$ because of the delay.

The proof for a vertex $u \in C_i \setminus R_0$ is similar.
Again we use that $f(u) = f'(C_j)$.
Let $u \in Y$, i.e., $d(u) = j'$.
If $C_j \setminus S = \emptyset$, then the set $Y$ is empty as well and there is nothing to prove for this case.
The clique $C_j$ was activated in the round $j'$ of the activation process ${\cal S}$.
Thus,
\[
  f(u) = f'(C_j) \leq |S_0 \cap C_j| + |S_{j' - 1} \cap Q| \leq |R_0 \cap C_i| + |R_{j' - 1} \cap Q| = a_{j'}^R(u).
\]
The vertex $u$ is activated in the round $j'$ of the activation process ${\cal R}$.
Let $u' \in C_i \setminus (X \cup Y)$, i.e., $d(u') = i'$.
%For the vertex $u'$ holds that $f'(C_i) = f(u')$.
The clique $C_i$ (included the vertex $u'$) was activated in the round $i'$ of the activation process ${\cal S}$.
Thus,
\[
  f(u') \leq a^S_{i'} = |S_{i' - 1} \cap Q| \leq |R_{i'-1} \cap Q| = a^R_{i'}(u').
\]
Thus, the vertex $u$ is activated in the round $d(u') = i'$, which finishes the proof of Claim~\ref{clm:NewTS}.
\end{proof}

\subsection{Neighborhood Diversity}
In this section we prove the \UNITSS problem admits an \FPT algorithm on graphs of the bounded neighborhood diversity.
We again use Lemma~\ref{lem:MAJtrueActivation}.
Note that in each round of the activation process at least one type has to be activated.
This implies that there are at most $\nd(G)$ rounds of the activation process.
We use this fact to model the whole activation process as an integer linear program which is then solved using Lenstra's celebrated result.

\begin{theorem}[\cite{lenstra83,FT87}]
Let $p$ be a number of integral variables in a mixed integer linear program and let $L$ be a number of bits needed to encode the program.
Then it is possible to find an optimal solution in time $\BigO(p^{2.5p}\poly(L))$ and a space polynomial in $L.$
\end{theorem}

There has to be an order in which the types are activated in order to activate the whole graph.
Since there are $t = \nd(G)$ types, we can try all such orderings.
Let us fix an order $\prec$ on types.
To construct the ILP we further need to know which types are fully activated at the beginning, since for these types we will not need to check there are enough neighbors activated in their round.
Denote by $c_0$ the number of such types.
Once the order $\prec$ is fixed the set of fully activated types at the beginning is determined by $c_0$.
Since $c_0$ can attain values $0,\ldots,t$ we can try all $t+1$ possibilities.
Now, with both $\prec$ and $c_0$ fixed, denote the set of the types activated in the beginning by $\mathcal{D}_0$.

Observe further that, as the vertices in a type share all neighbors, the only thing that matters is the number of activated vertices in each type and not the actual vertices activated.
Thus, we have variables $x_C$ which correspond to the number of vertices in type $C$ selected into a target set $S$.

Let $C$ be a type and $n_C$ be the number of vertices in $C$.
Since we know when $C$ is activated, we know how many active vertices are in $C$ in each round.
In $C$ there are $x_C$ vertices before the activation of $C$ and $n_C$ after the activation.
To formulate the integer linear program we denote the set of types by $\mathcal{D}$ and we write $D \in N(C)$ if the two corresponding vertices in the type graph $T_G$ are connected by an edge.

\begin{equation*}
\begin{aligned}
& \text{minimize}
& & \sum_{C\in \mathcal{D}} x_C & \\
& \text{subject to}
& & f'(C) \le \sum_{D\prec C, D \in N(C)}n_D + \sum_{D\succ C, D\in N(C)} x_D & \forall C\in \mathcal{D} \setminus \mathcal{D}_0 \mbox{ $C$ indep.}\\
& & & f'(C) \le x_C + \sum_{D\prec C, D \in N(C)}n_D + \sum_{D\succ C, D\in N(C)} x_D & \forall C\in \mathcal{D} \setminus \mathcal{D}_0 \mbox{ $C$ clique} \\
%& & & x_C = n_C & \forall C \in \mathcal{D}_0 \\
& \text{where}
& & 0 \leq x_C < n_C & \forall C \in \mathcal{D} \setminus \mathcal{D}_0 \\
& & & x_C = n_C & \forall C \in \mathcal{D}_0 \\
\end{aligned}
\end{equation*}
%The expression $[\mbox{$C$ is a clique}]$ has value $1$ if $C$ is a clique and has value zero otherwise.

As there are at most $t!$ orders of the set $\mathcal{D}$ and $t+1$ choices of~$c_0$, the \UNITSS problem can be solved in time ${(t+1)t! t^{\BigO(t)} \poly(n) = t^{\BigO(t)}\poly(n)}$.
Thus, we have proven Theorem~\ref{thm:MAJTSSisFPTwrtND}.

\section{Hardness}
In this section we prove hardness results for \TSS.
We prove the following problems are \Wh{1}.
\begin{enumerate}
 \item \TSS parameterized by the neighborhood diversity.
 \item \MAJTSS parameterized by the modular-width.
 \item \TSS parameterized by the twin cover number.
\end{enumerate}
Moreover unless ETH fails, there is no algorithm of running time $f(k)n^{o(k/\log k)}$ for the problems above, where $k$ is the appropriate parameter.
We use an \FPT reduction from the following problem.

\prob{\textsc{Colored Subgraph Isomorphism} (CSI)}
{A $k$-partite graph $G=(V_1 \dot\cup \cdots \dot\cup V_k,E)$, a graph $H$ such that $V(H) = \{1,\dots,k\}$.}
{Is $H$ isomorphic to a (colored) subgraph of $G$? I.e., is there an injective mapping \mbox{$\phi\colon V(H) \to V(G)$} such that $\{\phi(u),\phi(v)\} \in E(G)$ for each $\{u,v\} \in E(H)$ and $\phi(i) \in V_i$ for each $i \in V(H)$?}
If $(G,H)$ is a \emph{yes}-instance of CSI, then we say that \emph{$G$ contains a colored copy of $H$}.

\begin{theorem}[{Marx~\cite[Corollary 6.1]{Marx10}}]
  The CSI problem is \Wh{1} for the parameter tree width of~$H$.
  Moreover, if there is a recursively enumerable class $\mathcal{H}$ of graphs with unbounded treewidth, an algorithm $\mathbb{A}$, and an arbitrary function $f$ such that~$\mathbb{A}$ correctly decides every instance of CSI with the graph $H$ in $\mathcal{H}$ in time $f(k)\cdot|V(G)|^{o(\textup{tw}(H)/\log \textup{tw}(H))}$, then ETH fails.
\end{theorem}

It is known that there are infinitely many 3-regular graphs such that each such graph $H$ has treewidth $\Theta(k)$ (cf.~\cite[Proposition~1, Theorem~5]{GroheM09}).
Using the class of 3-regular graphs as $\mathcal{H}$ in the above theorem, we derive the following corollary.
\begin{corollary}\label{cor:psi_hard}
If there is an algorithm~$\mathbb{A}$ and a function $f$ such that~$\mathbb{A}$ correctly decides every instance of \textsc{Colored Subgraph Isomorphism} with the graph $H$ being 3-regular in time $f(k)\cdot|V(G)| ^{o(k/\log k)}$, then ETH fails.
\end{corollary}
Thus, from now we suppose the graph $H$ is 3-regular.
Let $(G,H)$ be an instance of CSI.
We will present three \FPT reductions to \TSS such that the sizes of the resulting graphs~$G'$ will be polynomial in $|V(G)|$ and the appropriate parameters will be bounded by $\BigO(k)$.
We refer to the sets $V_c$ as to \emph{color classes} of $G$ and to a set $E_{cd}$ as to edges between the color classes $V_c$ and $V_d$.
Note that we can suppose all color classes have the same size and $E_{cd}$ is not empty if and only if $\{c,d\} \in E(H)$.
In the CSI problem we need to select exactly one vertex from each color class $V_c$ and exactly one edge from each (nonempty) set~$E_{cd}$.
Moreover, we have to make certain that if ${\{u,v\} \in E_{cd}}$ is a selected edge, then ${u \in V_c}$ and ${v \in V_d}$ are the selected vertices.
This selection of vertices and edges and the checking of their incidence needs to be captured in the reductions to \TSS.

\subsection{Neighborhood Diversity}
For an easier notation during the reduction, we denote the size of an arbitrary color class $V_c$ by $n + 1$.
Similarly, we denote $m_{cd} + 1 = |E_{cd}|$.
Thus, we have $V_c = \{v^c_0,\dots,v^c_n\}$ and $E_{cd} = \{e^{cd}_0,\dots,e^{cd}_{m_{cd}}\}$.
We describe a reduction from the graph $G$ to an instance $(G', f, b)$ of \TSS such that $\nd(G')$ is $\BigO(k)$.

\subsubsection{Overview}
As the proof is quite long and technical we overview main ideas contained in the proof here.
We encode a vertex $v^c_i$ in a color class $V_c$ of the graph $G$ by two numbers $i$ and $n-i$.
The selection gadget for $V_c$ has two parts: {\up} and {\down}.
We select the $v^c_i \in V_c$ for the copy of $H$ if and only if we select to a target set $i$ vertices in the {\up} part of the gadget for $V_c$ and $n - i$ vertices in the {\down} part.
We proceed with encoding of edges similarly, however, edges are encoded by multiples of sufficiently large number~$q$.
I.e., we select the $\ell$-th edge in $E_{cd}$ for the copy of $H$ if and only if we select $q\ell$ in the {\up} part of the gadget for $E_{cd}$ and $q(m_{cd} - \ell)$ in the {\down} part.
We will create three types of gadgets: for selection vertices, for selection edges, and gadgets which check that the selected vertices are incident to the selected edges.

The incidence check is done by the following idea.
The incidence gadget for $V_c$ and $E_{cd}$ contains 2 vertices $u^{cd}(\ell)$ and $w^{cd}(\ell)$ for each edge $e^{cd}_\ell \in E_{cd}$.
The vertices $u^{cd}(\ell)$ are connected to the {\up} parts of the selection gadgets for $V_c$ and $E_{cd}$ and the vertices $w^{cd}(\ell)$ are connected to the {\down} parts.
Let $v^c_i$ be the vertex of $e^{cd}_\ell$ which is in $V^c$.
The threshold $f\bigl(u^{cd}(\ell)\bigr)$ is set to $i + q\ell$ and $f\bigl(w^{cd}(\ell)\bigr) = n - i + q(m_{cd} - \ell)$.

Now suppose we selected $v^c_j \in V_c$ and the edge $e^{cd}_\ell$ for the copy of $H$ and the vertices $u = u^{cd}(\ell)$ and $w = w^{cd}(\ell)$ were activated during the activation process in $G'$.
The vertex $u$ is connected  to $j$ selected neighbors in the {\up} part of the vertex selection gadget and to $q\ell$ selected neighbors in the {\up} part of the edge selection gadget.
Since $f(u) = i + q\ell$, it must holds that $j \geq i$.
Similarly, the vertex $w$ is connected to $n - j + q(m_{cd} - \ell)$ selected vertices in the {\down} parts of the selection gadgets and $f(w) = n - i + q(m_{cd} - \ell)$.
Thus, we have $j \leq i$.
Therefore, if we check that the vertices $u$ and $w$ are activated during the activation process we can verify that $i = j$ and the vertex $v^c_j$ is actually incident to the edge $e^{cd}_\ell$.
To make this incidence check actually work, we need that
\begin{align*}
 f\bigl(u^{cd}(\ell)\bigr) < f\bigl(u^{cd}(\ell')\bigr)  \Leftrightarrow \ell < \ell'
 \quad\text{and}\quad
 f\bigl(w^{cd}(\ell)\bigr) < f\bigl(w^{cd}(\ell')\bigr)  \Leftrightarrow \ell' < \ell.
\end{align*}
This can be done by setting $q$ to $2n$,  which is sufficiently large.
Now, we proceed with the formal description of the reduction.
We point out that all types of $G'$ in the neighborhood diversity decomposition are independent sets.

\subsubsection{Selection Gadget $L(s)$}
First, we describe gadgets of the graph $G'$ for selecting vertices and edges of the graph $G$.
For an overview of the gadget please refer to Figure~\ref{fig:NDSelectionGadget}.
The gadget $L(s)$ is formed by two types {$L$-\down} and {$L$-\up} of equal size $s$ (the number $s$ will be determined later); we refer to these two types as the \emph{selection part}.
For a vertex $v$ in the selection part we set the value $f(v)$ to the degree of $v$.
It means that if some vertex $v$ from the selection part is not selected into the target set, then all neighbors of  $v$ have to be active before the vertex $v$ can be activated by the activation process.
The selection gadget $L$ is connected to the rest of the graph using only vertices from the selection part.
The last part of the gadget $L$ is formed by the type {$L$-\guard} of  $s + 1$ vertices connected to both types in the selection parts.
For each vertex $v$ in {$L$-\guard} type we set $f(v) = s$.

\begin{figure}[bt]
\begin{minipage}{.45\textwidth}
    \usetikzlibrary{positioning,fit,calc}

\begin{tikzpicture}[node distance=2.5cm,font=\footnotesize]
  \tikzstyle{bag} = [circle, draw, minimum width=1.3cm]
  \tikzstyle{edge} = [black, ultra thick]
  \tikzstyle{fitting} = [draw, dashed, rounded corners, black!40, thick, inner sep=15pt]

  \node[bag, label={270:${\rm deg}$}, label={90:$L$-{\rm pos}}] (VaUP) {$s$};
  \node[bag, label={270:${\rm deg}$}, label={90:$L$-{\rm neg}}, below of=VaUP] (VaDOWN) {$s$};
  \node[bag, label={270:$s$}, label={90:$L$-{\rm guard}}] (VaCTRL) at ($(VaUP)!.5!(VaDOWN) - (2,0)$) {$s + 1$};

  \draw[edge] (VaUP) -- (VaCTRL) -- (VaDOWN);

  \node[right of=VaUP] (dummyUP) {};
  \node[right of=VaDOWN] (dummyDOWN) {};

  \draw[edge] (VaUP) -- (dummyUP);
  \draw[edge] (VaDOWN) -- (dummyDOWN);

  \node[fitting, fit={(VaUP) (VaDOWN)}, label={[yshift=-2pt]90:Selection part}, xshift=10pt] {};
\end{tikzpicture}
\end{minipage}
\begin{minipage}{.54\textwidth}
  \caption{An overview of the selection gadget $L(s)$. Numbers in circles denote numbers of vertices in each type and numbers under circles denote thresholds of vertices in each type.}
  \label{fig:NDSelectionGadget}
\end{minipage}
\end{figure}

\begin{lemma}\label{lem:SelectionGadgetProperties}
Suppose there is a selection gadget $L(s)$ in the graph $G'$.
We claim that exactly $s$ vertices of the gadget $L$ are needed to be selected in the target set $S$ to activate the vertices in the {$L$-\guard} type.
Moreover, these $s$ vertices have to be selected from the selection part of $L$.
\end{lemma}
\begin{proof}
Let $S' = V(L) \cap S$, i.e., vertices of the target set $S$ in the gadget $L$.
First, suppose $|S'| < s$ or $|S'| = s$ and some vertex $u$ of {$L$-\guard} is in $S'$.
Since there are $s + 1$ vertices in $L$-\guard, there is a vertex $v$ in the {$L$-\guard} type such that $v \notin S'$.
Let $V^p$ be vertices of the selection part of $L$.
The vertex $v$ has neighbors only in $V^p$ and the threshold of $v$ is $s$.
Note that $|V^p \cap S'| < s$.
Thus, at least one vertex $w \in V^p \setminus S'$ needs to be activated during the activation process before the vertex $v$ is activated.
However, $f(w) = \deg(w)$ for $w \in V^p$.
Therefore, the vertex $w$ has to be activated after the vertex $v$ is activated.
That is a contradiction and $|V^p \cap S'| \geq s$ must hold.
When $S'$ contains $s$ vertices from the selection part of $L$, then it is easy to see that the all vertices in {$L$-\guard} type are activated in the first round of the activation process.
\end{proof}

\subparagraph{Numeration of Vertices and Edges}
%Let ${V_c = \{v_0,\dots,v_n\}}$.
%By Lemma~\ref{lem:SelectionGadgetProperties}, we can encode selection of vertices and edges of the graph $G$ to the colored copy of $H$.
For every color class $V_c$ we create a selection gadget $L_c = L(n)$.
We select a vertex $v^c_i \in V_c$ to the colored copy of $H$ if $i$ vertices in the {$L_c$-\up} type and $n-i$ vertices in the {$L_c$-\down} type of the gadget $L_c$ are selected into the target set.

The selection of edges is similar, however, a bit more complicated.
%Let $q \in \N$ and $E_{cd} = \{e_0,\dots,e_{m_{cd}}\}$.
For every set $E_{cd}$ we create a selection gadget $L_{cd}$ of kind $L(qm_{cd})$.
We select an edge $e^{cd}_j \in E_{cd}$ to the colored copy of $H$ if $qj$ vertices in the {$L_{cd}$-\up} type of the gadget $L_{cd}$ are selected into the target set (and $q(m_{cd}-j)$ vertices in the {$L_{cd}$-\down} are selected into the target set).
Suppose $s$ vertices in the {$L_{cd}$-\up} type are selected into the target set.
If $s$ is not divisible by $q$, then it is an invalid selection.
We introduce a new gadget such that $s$ has to be divisible by $q$.

\subsubsection{Multiple Gadget $M(q,s)$}
A multiple gadget $M(q,s)$ consists of a selection gadget $L(qs)$ and three other types: {$M$-\up}, {$M$-\down} of $s$ vertices, and {$M$-\guard} of $qs$ vertices.
The type {$M$-\up} is connected to the type {$L$-\up} and the type {$M$-\down} is connected to the type {$L$-\down}.
The type {$M$-\guard} is connected to the types {$M$-\up} and $M$-\down.
Still, the rest of graph $G'$ is connected only to the types {$L$-\up} and {$L$-\down}.
Let $\{u_1,\dots,u_s\}$ and $\{w_1,\dots,w_s\}$ be vertices in {$M$-\up} type and {$M$-\down} type, respectively.
We set thresholds $f(u_i) = f(w_i) = qi$.
For each vertex $v$ in {$M$-\guard} we set $f(v) = s$.
For an example of multiple gadget see Figure~\ref{fig:NDMultipleGadget}.

\begin{figure}[bt]
\begin{minipage}{.8\textwidth}
    \usetikzlibrary{positioning,fit,calc}

\begin{tikzpicture}[node distance=2.5cm,font=\footnotesize]
  \tikzstyle{bag} = [circle, draw, minimum width=1.3cm]
  \tikzstyle{edge} = [black, ultra thick]
  \tikzstyle{fitting} = [draw, dashed, rounded corners, black!40, thick, inner sep=15pt]

  \node[bag, label={90:$L$-pos}, label={270:${\rm deg}$}] (EabUP) {$qs$};
  \node[bag, label={90:$L$-guard}, label={270:$qs$}, below right of=EabUP] (EabCTRL) {$qs + 1$};
  \node[bag, label={90:$L$-neg}, label={270:${\rm deg}$}, below left of=EabCTRL] (EabDOWN) {$qs$};

  \node[bag, label={90:$M$-pos}, label={270:$q,\dots,qs$}, above right of=EabCTRL] (EabUPtimesN) {$s$};
  \node[bag, label={90:$M$-neg}, label={270:$q,\dots,qs$}, below right of=EabCTRL] (EabDOWNtimesN) {$s$};
  \node[bag, label={90:$M$-guard}, label={270:$s$}, below right of=EabUPtimesN] (EabCTRLtimesN) {$qs$};

  \node[fitting, fit={(EabDOWN) (EabUP) (EabCTRL)}, xshift=-10pt, label={[yshift=-18pt, xshift=-3pt]70:$L(qs)$}] (Eab) {};

  \node[fitting, fit={(Eab)(EabUPtimesN)(EabDOWNtimesN)(EabCTRLtimesN)}, label={[yshift=-18pt, xshift=40pt]70:$M(q,s)$}, inner sep=5pt] {};

  \node[left of=EabUP] (dummyUP) {};
  \node[left of=EabDOWN] (dummyDOWN) {};

  \draw[edge] (EabUP) -- (EabCTRL) -- (EabDOWN) -- (EabDOWNtimesN) -- (EabCTRLtimesN) -- (EabUPtimesN) -- (EabUP);

  \draw[edge] (EabUP) -- (dummyUP);
  \draw[edge] (EabDOWN) -- (dummyDOWN);
\end{tikzpicture}
\end{minipage}
\begin{minipage}{.18\textwidth}
  \caption{An overview of the multiple gadget $M(q,s)$.}
  \label{fig:NDMultipleGadget}
\end{minipage}
\end{figure}

\begin{lemma}
\label{lem:MultipleGadget}
Suppose there is a multiple gadget $M(q,s)$ in the graph $G'$.
Let $L$ be a selection gadget in $M$.
We claim that exactly $qs$ vertices of the gadget $L$ are needed to be selected in the target set $S$ to activate the types {$L$-\guard}, {$M$-\up}, {$M$-\down} and {$M$-\guard}.
Moreover, these $qs$ vertices have to be selected from the selection part of $L$ and the numbers of vertices selected in {$L$-\up} and {$L$-\down} types are divisible by $q$.
\end{lemma}
\begin{proof}
By Lemma~\ref{lem:SelectionGadgetProperties}, we know that $qs$ selected vertices in the types {$L$-\up} and {$L$-\down} are needed to activate the {$L$-\guard} type.
Suppose there are $z$ vertices in the {$L$-\up} type selected into a target set and $z$ is not divisible by $q$.
It follows there are $qs - z$ selected vertices in {$L$-\down}.
Thus, $z = qa + r, 0 < r < q$ and $qs - z = q(s-a) - r$.
Let $\{u_1,\dots,u_s\}$ and $\{w_1,\dots,w_s\}$ be vertices in the {$M$-\up} type and in the {$M$-\down} type, respectively.
Recall that $f(u_i) = f(w_i) = qi$.
Thus, vertices $u_1,\dots,u_a$ and $w_1,\dots,w_{s-a-1}$ are activated in the first round of the activation process.

We claim that no other vertices in the gadget $M$ would be activated during the activation process.
Vertices in the {$M$-\guard} type have only $s-1$ activated vertices among their neighbors and have the threshold $s$.
Vertices in {$L$-\up} and {$L$-\down} have thresholds equal to their degrees.
Thus, they have to be activated after all vertices in {$M$-\up} and {$M$-\down} are activated.
Vertices $u_{a+1},\dots,u_s$ in the {$M$-\up} type and $w_{s-a},\dots,w_s$ in the {$M$-\down} type cannot be activated unless some of their neighbors are activated.

Now suppose that $r = 0$, i.e., $z = qa$ and $qs - z = q(s -a)$.
Vertices $u_1,\dots,u_a$ and $w_1,\dots,w_{s-a}$ are activated in the first round.
All vertices in the {$M$-\guard} type are activated in the second round because they have $s$ activated vertices among their neighbors.
Recall that the maximum threshold in the {$M$-\up} and the {$M$-\down} type is $qs$.
Since there are $qs$ vertices in $M$-\guard, every vertex in the types {$M$-\up} and {$M$-\down} has at least $qs$ activated vertices among its neighbors.
Therefore, the rest of vertices in the types {$M$-\up} and {$M$-\down} are activated in the third round.
\end{proof}

\subsubsection{Incidence Gadget}
So far we described how we encode in graph $G'$ selecting vertices and edges to the colored copy of $H$.
It remains to describe how we encode the correct selection, i.e., if $v \in V_c$ and $e \in E_{cd}$ are selected vertex and edge to the colored copy of $H$, then $v \in e$.
We create a selection gadget $L_{c}(n)$ for each color class $V_c$.
We set the number $q$ to $2n$ and create a multiple gadget $M_{cd}$ of kind $M(2n,m_{cd})$ (with selection gadget $L_{cd}$) for each set $E_{cd}$.
We join gadgets $L_c$ and $M_{cd}$ through an incidence gadget $I_{c:cd}$.
See Figure~\ref{fig:NDTSSHardness}, for better understanding how the incidence gadget is connected to the selection and multiple gadgets.
The incidence gadget $I_{c:cd}$ has three types {$I_{c:cd}$-\up} and {$I_{c:cd}$-\down} of $m_{cd} + 1$ vertices each and {$I_{c:cd}$-\guard} of $n + 2nm_{cd}$ vertices.
We connect the {$I_{c:cd}$-\guard} type to the types {$I_{c:cd}$-\up} and {$I_{c:cd}$-\down}.
Furthermore, we connect the type {$I_{c:cd}$-\up} to the types {$L_c$-\up} and {$L_{cd}$-\up}.
Similarly, we connect the type {$I_{c:cd}$-\down} to the types {$L_c$-\down} and {$L_{cd}$-\down}.

We set thresholds of all vertices in the {$I_{c:cd}$-\guard} type to $m_{cd} + 2$.
Recall there are $m_{cd} + 1$ edges in the set $E_{cd}$.
Thus, we can associate edges in $E_{cd}$ with vertices in {$I_{c:cd}$-\up} ({$I_{c:cd}$-\down} respectively) one-to-one.
I.e., $V(I_{c:cd}\text{-\up}) = \setof{u_\ell}{e_\ell \in E_{cd}}$ and $V(I_{c:cd}\text{-\down}) = \setof{w_\ell}{e_\ell \in E_{cd}}$.
Let $v^c_i \in V_c$ and $e^{cd}_j \in E_{cd}$ with $v^c_i \in e^{cd}_j$.
Recall that selecting $v^c_i$ and $e^{cd}_j$ into the colored copy of $H$ is encoded as selecting $i$ vertices in {$L_c$-\up} type and $2nj$ vertices in {$L_{cd}$-\up} type into the target set.
We set threshold of $u_{j}$ to $i + 2nj$ and threshold of $w_{j}$ to the ``opposite'' value $n - i + 2n(m_{cd} -j)$.

\begin{figure}[bt]
  \begin{center}
    \usetikzlibrary{positioning,fit,calc}

\begin{tikzpicture}[node distance=2.2cm, font=\scriptsize]
  \tikzstyle{bag} = [circle, draw, minimum width=1cm, inner sep=1pt]
  \tikzstyle{emptyBag} = [circle, minimum width=1cm, inner sep=1pt]
  \tikzstyle{edge} = [black, ultra thick]
  \tikzstyle{fitting} = [draw, dashed, rounded corners, black!40, thick, inner sep=8pt]
  \tikzstyle{fittingTesny} = [draw, dashed, rounded corners, black!40, thick, inner sep=4pt]
  \tikzstyle{threshold} = [orange]

  \node[bag, label={270:$L_c\mbox{-pos}$}, label={[threshold]180:$\deg$}] (VaUP) {$n$};
  \node[bag, label={[xshift=-2pt]270:$L_c\mbox{-guard}$}, label={[threshold]0:$n$}, below left of=VaUP] (VaCTRL) {$n + 1$};
  \node[bag, label={270:$L_c\mbox{-neg}$}, label={[threshold]180:$\deg$}, below right of=VaCTRL] (VaDOWN) {$n$};
  \node[emptyBag] at ($(VaDOWN.south)$) (VaE) {};
  \node [fittingTesny, fit={(VaUP) (VaDOWN) (VaCTRL) (VaE)}, label={[yshift=-18pt,xshift=3pt]110:$L_c$}, xshift=-2pt] {};

  \draw[edge] (VaUP) -- (VaCTRL) -- (VaDOWN);

  \node[bag, label={[xshift=-4pt]270:$I_{c:cd}\mbox{-pos}$}, label={[threshold, yshift=-5pt]135:inc}] at ($(VaUP) + (1.65cm,0)$) (IaabUP) {$m_{cd}+1$};
  \node[bag, label={270:$I_{c:cd}\mbox{-guard}$}, label={[threshold]180:$m_{cd}+2$}, below right of=IaabUP] (IaabCTRL) {$n + 2nm_{cd}$};
  \node[bag, label={[xshift=6pt]270:$I_{c:cd}\mbox{-neg}$}, label={[threshold, yshift=-5pt]135:inc}, below left of=IaabCTRL] (IaabDOWN) {$m_{cd}+1$};
  \node[emptyBag] at ($(IaabUP)!.5!(IaabDOWN) - (.4cm, 0)$) (IaabE) {};
  \node[emptyBag] at ($(IaabDOWN.south)$) (IaabE2) {};
  \node [fittingTesny, fit={(IaabUP) (IaabDOWN) (IaabCTRL) (IaabE) (IaabE2)}, label={[xshift=37pt,yshift=-18pt]90:$I_{c:cd}$}, xshift=3pt] {};

  \node[bag, label={270:$L_{cd}\mbox{-pos}$}, right= of IaabUP, label={[threshold, yshift=6pt]0:$\deg$}] (EabUP) {$2nm_{cd}$};
  \node[bag, label={270:$L_{cd}\mbox{-guard}$}, below right of=EabUP, label={[threshold]180:$2nm_{cd}$}] (EabCTRL) {$2nm_{cd} + 1$};
  \node[bag, label={270:$L_{cd}\mbox{-neg}$}, below left of=EabCTRL, label={[threshold, yshift=6pt]0:$\deg$}] (EabDOWN) {$2nm_{cd}$};
%  \node[fitting, fit={(EabUP) (EabDOWN) (EabCTRL)}, label={[xshift=38pt,yshift=-18pt]90:$L_{cd}$}, yshift=-5pt, xshift=2pt] (EabSEL) {};

  \draw[edge] (IaabUP) -- (IaabCTRL) -- (IaabDOWN);

  \node[bag, label={[xshift=-8pt]270:$M_{cd}\mbox{-pos}$}, right=of EabUP, label={[threshold]0:$[m_{cd}]\left\{2n\right\}$}] (EabUPtimesN) {$m_{cd}$};
  \node[bag, label={[xshift=5pt]270:$M_{cd}\mbox{-guard}$}, below right of=EabUPtimesN, label={[threshold]180:$m_{cd}$}] (EabCTRLtimesN) {$2nm_{cd}$};
  \node[bag, label={[xshift=-8pt]270:$M_{cd}\mbox{-neg}$}, below left of=EabCTRLtimesN, label={[threshold]0:$[m_{cd}]\left\{2n\right\}$}] (EabDOWNtimesN) {$m_{cd}$};

  \node[fitting, fit={(EabDOWNtimesN) (EabUPtimesN) (EabCTRLtimesN) (EabUP)}, xshift=6pt, yshift=-6pt, label={[yshift=-12pt, xshift=2pt]45:$M_{cd}$}] (Eab) {};

% \node[fitting, fit={(EabDOWN) (EabUP) (EabCTRL) }, xshift=-1pt, label={[yshift=-18pt, xshift=-45pt]110:$M_{cd}$}] {};

  \draw[edge] (EabUP) -- (EabCTRL) -- (EabDOWN) -- (EabDOWNtimesN) -- (EabCTRLtimesN) -- (EabUPtimesN) -- (EabUP);

  \draw[edge] (VaUP) -- (IaabUP) -- (EabUP);

  \draw[edge] (VaDOWN) -- (IaabDOWN) -- (EabDOWN);
\end{tikzpicture}
  \end{center}
  %\includestandalone[scale=.75]{tikz/NDTSSgadget}
  \caption{An~overview of the reduction.
  The number inside a type is the number of vertices of the type.
  The threshold of vertices in a type is displayed next to the type in orange (light-gray).
  }
  \label{fig:NDTSSHardness}
\end{figure}

Since we set the coefficient $q$ to $2n$, for each edge $e_j \in E_{cd}$ and each vertex $v_i \in V_c$ the sum $i + 2nj$ is unique.
Thus, every vertex in {$I_{c:cd}$-\up} ({$I_{c:cd}$-\down}) has a unique threshold.
We will use this number to check the incidence.

%Note further that if we have a total order on edges $\preceq$ using the order on numbers $m_e$ then fulfilling edge threshold for an edge $e$ implies that the threshold is fulfilled for all edges $e'$ with ${e'\preceq e}$.

We described how from the graph $G$ with $k$ color classes (instance of CSI) we create the graph $G'$ with the threshold function $f$ (input for \TSS).
%For every color class $V_c$ we create a selection gadget $L_c$.
%For every edge set $E_{cd}$ we create a multiple gadget $M_{cd}$.
We join the gadgets $L_c$ and $M_{cd}$ by an incidence gadget $I_{c:cd}$ (gadgets $L_d$ and $M_{cd}$ are joint by a gadget $I_{d:cd}$).
Let $\bar{m} = \sum_{\{c,d\} \in E(H)}m_{cd}$.
To finish the construction of the instance of \TSS, we set the budget $b$ to $kn + 2n\bar{m}$.

\subsubsection{Finishing the Reduction}
It is easy to see the following observations follow directly from the construction of $G'$.

\begin{observation}
\label{obs:HardnessNDSize}
The graph $G'$ has polynomial size in the size of the graph $G$ and can be constructed in time polynomial in~$|V(G)|$.
\end{observation}

\begin{observation}
\label{obs:HardnessNDParameter}
The neighborhood diversity of the graph $G'$ is $\BigO(k)$.
\end{observation}

\begin{theorem}
\label{thm:HardnessND1}
If the graph $G$ contains a colored copy of $H$, then $G'$ with the threshold function~$f$ contains a target set of size $b$.
\end{theorem}
\begin{proof}
Let $K$ be a colored copy of $H$ in the graph $G$.
We construct a set $S \subseteq V(G')$.
Let $v^c_i \in V(K) \cap V_c$.
We add $i$ vertices in the {$L_c$-\up} type and $n - i$ in the {$L_c$-\down} type into the set $S$.
Let $e^{cd}_j \in E(K) \cap E_{cd}$.
For the set $E_{cd}$ we have a multiple gadget $M_{cd}$ and there is a selection gadget $L_{cd}$ inside $M_{cd}$.
We add $2nj$ vertices in the {$L_{cd}$-\up} type and $2n(m_{cd}-j)$ vertices in the {$L_{cd}$-\down} into the set $S$.
We have $n$ vertices in $S$ for every color class $V_c$ and $2nm_{cd}$ vertices in $S$ for every edge set $E_{cd}$.
Thus,
\(
  |S| = kn + 2n\bar{m} = b
\).

We claim the set $S$ is a target set.
We analyze the selection gadget $L_c$, the multiple gadget $M_{cd}$ (with the $L_{cd}$ selection gadget), and the incidence gadget $I_{c:cd}$.
All vertices in the types {$L_c$-\guard} and {$L_{cd}$-\guard} are activated in the first round (see the proof of Lemma~\ref{lem:SelectionGadgetProperties}).
All vertices in the types {$M_{cd}$-\down}, {$M_{cd}$-\up} and {$M_{cd}$-\guard} are activated during the first three rounds -- for details see the proof of Lemma~\ref{lem:MultipleGadget}.

Recall $V(I_{c:cd}\text{-\up}) = \setof{u_\ell}{e_\ell \in E_{cd}}$ and $V(I_{c:cd}\text{-\down}) = \setof{w_\ell}{e_\ell \in E_{cd}}$.
The threshold of $u_\ell \in V(I_{c:cd}\text{-\up})$ is $2n\ell + \ell'$ for some $\ell' \in \{0,\dots,n\}$.
There are $2nj + i$ vertices activated in the types {$L_{cd}$-\up} and {$L_c$-\up}.
Vertices $u_0,\dots,u_{j-1}$ are activated in the first round because their thresholds are strictly smaller than $2nj$.
The threshold of $u_{j}$ is $2nj + i$ because this vertex corresponds to the incidence $v_i \in e_j$.
Thus, the vertex $u_j$ is activated in the first round as well.
Vertices $u_{j+1},\dots,u_{m_{cd}}$ have thresholds bigger than $2n(j + 1)$ and cannot be activated in the first round.
By the same analysis we get that vertices $w_{j},\dots,w_{m_{cd}}$ in the {$I_{c:cd}$-\down} type are activated in the first round.

In the first round there are $m_{cd}+2$ activated vertices in the types {$I_{c:cd}$-\up} and {$I_{c:cd}$-\down} altogether.
All vertices in the {$I_{c:cd}$-\guard} type are activated in the second round because they have threshold $m_{cd} + 2$.
The maximum threshold in the {$I_{c:cd}$-\up} ({$I_{c:cd}$-\down}) type is $n + 2nm_{cd}$.
Thus, the rest of vertices in the types {$I_{c:cd}$-\up} and {$I_{c:cd}$-\down} are activated in the third round because they have $n + 2nm_{cd}$ active neighbors in the {$I_{c:cd}$-\guard} type.

All vertices outside the types {$L_{c}$-\up}, {$L_{c}$-\down}, {$L_{cd}$-\up}, and {$L_{cd}$-\down} are activated during the first three rounds.
Let $U$ be a set of vertices which are not activated during the first three rounds.
Note that $U$ is an independent set and for every $u \in U$ holds that $f(u) = \deg(u)$.
Therefore, vertices in $U$ are activated in the fourth round.
\end{proof}

\begin{theorem}
\label{thm:HardnessND2}
If the graph $G'$ with the threshold function $f$ contains a target set of size $b$, then $G$ contains a colored copy of $H$.
\end{theorem}
\begin{proof}
Let $S$ be a target set of the graph $G$ of size $b$.
There are $k$ selection gadgets $L(n)$ in $G'$.
By Lemma~\ref{lem:SelectionGadgetProperties}, the set $S$ has to contain at least $n$ vertices in the selection part of every gadget $L(n)$.
There is also a selection gadget $L(2nm_{cd})$ in multiple gadgets $M_{cd}$ for each $\{c,d\} \in E(H)$.
By Lemma~\ref{lem:MultipleGadget}, the set $S$ has to contain at least $2nm_{cd}$ vertices in the selection part of every gadget $L(2nm_{cd})$ in $M_{cd}$.
Since $|S| = b = kn + 2n\bar{m}$, there is no other vertex in $S$.

Now, for every $V_c$ and $E_{cd}$ we select a vertex (or an edge, respectively).
We select a vertex $v^c_i \in V_c$ if $|V(L_c\text{-\up}) \cap S| = i$.
We select an edge $e^{cd}_j \in E_{cd}$ if $|V(L_{cd})\text{-\up} \cap S| = 2nj$.
By Lemma~\ref{lem:SelectionGadgetProperties} and~\ref{lem:MultipleGadget}, we know that this selection is well-defined.
We claim that if $v = v^c_i \in V_c$ is the selected vertex and $e = e^{cd}_j \in E_{cd}$ is the selected edge, then $v \in e$.

For a contradiction suppose $v \notin e$.
We analyze the incidence gadget $I_{c:cd}$.
Let $V(I_{c:cd}\text-\up) = \{u_0,\dots,u_{m_{cd}}\}$ and $V(I_{c:cd}\text-\down) = \{w_0,\dots,w_{m_{cd}}\}$.
Vertices in the type {$I_{c:cd}$-\up} have $i + 2nj$ active neighbors.
Vertices in the type {$I_{c:cd}$-\down} have $n - i + 2n(m_{cd} - j)$ active neighbors.
As we say in the proof of the previous theorem, %Theorem~\ref{thm:HardnessND1},
vertices $u_0,\dots,u_{j - 1}$ and $w_{j+1},\dots,w_{m_{cd}}$ are activated in the first round and vertices $u_{j+1},\dots,u_{m_{cd}}$ and $w_{0},\dots,w_{j-1}$ are not activated.

It remains to analyze the vertices $u_j$ and $w_{j}$.
Suppose $u_j$ is activated in the first round.
Thus, $f(u_j) = i' + 2nj < i + 2nj$.
Note that $i' < i$ because we suppose $v \notin e$.
For the threshold of $w_{j}$ holds
\[
  f(w_{j}) = n - i' + 2n(m_{cd}-j) > n - i + 2n(m_{cd}-j).
\]
Since the vertex $w_{j}$ has $n - i + 2n(m_{cd}-j)$ active neighbors (i.e., $|N(w_j) \cap S_0| = n - i + 2n(m_{cd}-j)$), the vertex $w_{j}$ cannot be activated in the first round.
Thus, at least one of the vertices $u_j, w_{j}$ is not activated in the first round.

Any vertex of the type {$I_{c:cd}$-\guard} cannot be activated in the second round because they have threshold $m_{cd} + 2$ and they have at most $m_{cd} + 1$ active neighbors in the first round.
Vertices of the type {$I_{c:cd}$-\guard} have to be activated after some other vertices in the types {$I_{c:cd}$-\up} or {$I_{c:cd}$-\down} are activated.
Let $W$ be the neighbors of {$I_{c:cd}$-\up} or {$I_{c:cd}$-\down} that are not from {$I_{c:cd}$-\guard}.
By the construction, the vertices of $W$ are only in the selection parts of the  gadgets.
Thus, for every $w \in W$ holds that $f(w) = \deg(w)$.
There are no vertices in the neighborhood of {$I_{c:cd}$-\up} and {$I_{c:cd}$-\down} which would be activated during the activation process and the vertices of the type {$I_{c:cd}$-\guard} stay inactive during the whole process.
Therefore, $S$ is not a target set, which is a contradiction.
\end{proof}

Theorem~\ref{thm:TSSisWwrtND} is a corollary of Theorem~\ref{thm:HardnessND1} and~\ref{thm:HardnessND2} and Observation~\ref{obs:HardnessNDSize} and~\ref{obs:HardnessNDParameter}.

\subsection{Modular-width}
%\subparagraph{Overview of Proof of Theorem~\ref{thm:MAJTSSisWwrtMW}}

In fact this can be seen as a~clever twist of the~ideas contained in the~proof of Theorem~\ref{thm:TSSisWwrtND}.
There are some nodes of the~neighborhood diversity decomposition already operating in the~majority mode -- e.g. guard vertices -- these we keep untouched.
For vertices with threshold set to their degree one has to ``double'' the~number of vertices in the~neighborhood. % and make sure that no newly added vertex is activated before all with threshold $\deg$.
Finally, one has to deal with types having different thresholds for each of its vertices, which is quite technical.
Here we exploit the~property of the~previous proof -- that these vertices naturally come in pairs and that it is possible to replace each of these vertices by a~collection of cliques.
This ensures that even if the~neighborhood is the~same some vertices get activated and some not.

As already noted, we reduce the instance $(G,H)$ of CSI into an instance~$(G',b)$ of \MAJTSS.
The graph~$G'$ consists of several different types of gadgets.
Throughout the~description of gadgets, we use the~following terminology.
The~vertices that are adjacent to vertices outside of the gadget (that introduces them) are referred to as \emph{interface vertices}, the~other vertices of the gadget are \emph{internal vertices}.
The~vertices outside of the gadget that are adjacent to its interface part are \emph{neighboring vertices of the~gadget}.
As we try to adapt the previous reduction to work in a ``majority'' regime, we need to know how many neighboring vertices gadgets have, to make gadgets work properly.

\subsubsection{Selection Gadget $L(s,r)$}
This gadget (together with an appropriate budget value) ensures that every target set encodes a~number from $0$ to $s$.
The~parameter $r$ denotes the~number of neighboring vertices, that is, we assume that each interface vertex has exactly~$r$ neighbors not in the gadget.
The~gadget consists of the following five types (see Figure~\ref{fig:selGadgetMajTSS} for an overview of the construction):
\begin{itemize}
	\item Two independent sets each of size $s$ called $L$-pos and $L$-neg (selection part of the gadget).
	\item One independent set of size $3s$ called $L$-guard.
	\item One independent set of size $r+3s$ called $L$-doubling.
	\item One independent set of size $2s$ called $L$-end.
\end{itemize}
The~$L$-guard type is connected to $L$-pos and $L$-neg, $L$-doubling is connected to $L$-pos, $L$-neg, and $L$-end.
The~$L$-pos and the $L$-neg types form the~interface part of the~gadget.
\begin{figure}[bt]
	\usetikzlibrary{positioning,fit,calc,decorations.pathreplacing}

\begin{tikzpicture}[node distance=2.3cm]
  \tikzstyle{bag} = [circle, draw, minimum width=1.3cm]
  \tikzstyle{edge} = [black, ultra thick]
  \tikzstyle{dummy} = [circle, dotted, draw, minimum width=1.3cm]
  \tikzstyle{fitting} = [draw, dashed, rounded corners, black!40, thick, inner sep=9pt]

  \node[bag, label={90:$L$-{\rm pos}}] (VaUP) {$s$};
  \node[bag, label={90:$L$-{\rm neg}}, below of=VaUP] (VaDOWN) {$s$};
  \node[bag, label={[yshift=-2pt]90:$L$-{\rm guard}}] at ($(VaUP)!.5!(VaDOWN) - (2.5,0)$) (VaCTRL) {$3s$};

  \draw[edge] (VaUP) -- (VaCTRL) -- (VaDOWN);

  %% Dummy nodes in the neighborhood
  \node[dummy, right of=VaUP] (dummyUP) {$r$};
  \node[dummy, right of=VaDOWN] (dummyDOWN) {$r$};

  \draw[edge] (VaUP) -- (dummyUP);
  \draw[edge] (VaDOWN) -- (dummyDOWN);

  \node[bag, left of=VaCTRL, label={90:$L$-{\rm doubling}}] (VaLARGE) {$3s + r$};
  \node[bag, left of=VaLARGE, label={90:$L$-{\rm end}}] (VaFinisher) {$2s$};

  \draw[edge] (VaDOWN) [out=180,in=315] to (VaLARGE);
  \draw[edge] (VaLARGE) [out=45,in=180] to (VaUP);
  \draw[edge] (VaLARGE) -- (VaFinisher);

  \node[fitting, fit={(VaUP) (VaDOWN)}, label={[yshift=-3pt]90:Selection part},yshift=6pt] {};

  \draw [decorate,decoration={brace,amplitude=10pt}] ($(VaCTRL) - (0,3) + (1,0)$) -- ($(VaFinisher) - (.6,3)$) node [black,midway,yshift=-15pt] {internal};
  \draw[dotted] ($(VaCTRL) - (0,3) + (1,0)$) -- ($(VaCTRL) + (0,3) + (1,0)$);

  \draw [decorate,decoration={brace,amplitude=10pt}] ($(VaCTRL) - (0,3) + (4,0)$) -- ($(VaCTRL) - (0,3) + (1,0)$) node [black,midway,yshift=-15pt] {interface};
  \draw[dotted] ($(VaCTRL) - (0,3) + (4,0)$) -- ($(VaCTRL) + (0,3) + (4,0)$);

  \draw [decorate,decoration={brace,amplitude=10pt}] ($(VaCTRL) - (0,3) + (6,0)$) -- ($(VaCTRL) - (0,3) + (4,0)$) node [black,midway,yshift=-15pt] {neighbors};
\end{tikzpicture}
	\caption{An~overview of the~selection gadget $L(s,r)$ for \MAJTSS}
	\label{fig:selGadgetMajTSS}
\end{figure}
The next definition gives us a useful tool for showing a lower bound on a number of vertices of a gadget that need to be in every target set.

\begin{definition}
	\label{def:deficit}
	Let $J$ be a gadget in $G'$ and let $\ol{V_{J}}$ be the set $V(G') \setminus V(J)$.
  The \emph{deficit $\defic_J(v)$ of a vertex $v \in V(J)$ with respect to $J$} is defined as
	\[
    \defic_J(v) = \begin{cases}
  		0 & \mbox{if } |N(v) \cap \ol{V_{J}}| \geq \deg(v) / 2 \\
  		\lceil \deg(v)/2 \rceil - |N(v) \cap \ol{V_J}| & \mbox{otherwise}.
  	\end{cases}
  \]
	We usually say just \emph{the deficit of a vertex~$v$} when the gadget~$J$ is clear from the context.
\end{definition}

In other words, deficit determines how many neighbors of~$v$ in $J$ we need to activate, provided that every vertex outside of $J$ is already activated.
Note that for an internal vertex, its deficit is exactly its threshold.
For an interface vertex, its deficit is its threshold decreased by a number of adjacent vertices outside $J$.

The following lemma gives us a lower bound on how many elements of $J$ need to be in any target set.
\begin{lemma}
	\label{lem:deficitLowerBound}
	If all vertices in a gadget $J$ have deficit at least $d$ and $S \subseteq V(G')$ is a set with $|S \cap V(J)| < d$, then $S$ is not a target set.
\end{lemma}
\begin{proof}
	We show a stronger statement that $S' = S \cup \ol{V_J}$ is not a target set.
  We see that every vertex $v \in V(J)$ has at most $|S \cap V(J)| + |N(v) \cap \ol{V_J}|$ neighbors in $S'$; which is at most $d - 1 + |N(v) \cap \ol{V_J}|$.
  However, from the definition of deficit and the assumption that deficit is at least $d$ we have $\lceil \deg(v)/2 \rceil \geq |N(v) \cap \ol{V_J}| + d$.
  Thus, no vertex in~$J$ has enough active neighbors and the activation process for~$S'$ terminates immediately.
\end{proof}
The following is then a simple consequence.
\begin{lemma}
	\label{lem:MWSelGadgetMin}
	Let $L = L(s,r)$ be a~selection gadget in the graph $G'$.
  If $S \subseteq V(G')$ contains less than $s$ vertices in $V(L)$, then $S$ is not a~target set.
\end{lemma}
\begin{proof}
	We examine the~values of deficits for vertices in gadget $L$.
  The~degrees of vertices in $L$-end, $L$-doubling, and $L$-guard are $3s+r$, $4s$, and $2s$, respectively, and therefore their deficits are $(3s+r)/2$, $2s$, and $s$, respectively.
  The~vertices in $L$-pos and $L$-neg have degree $6s+2r$ and so their threshold is $3s+r$.
  As they have $r$ neighbors in $\ol{V_L}$, their deficit is $3s$.
	We see that every vertex has the deficit at least $s$.
  The lemma follows from Lemma~\ref{lem:deficitLowerBound}.
\end{proof}
\begin{lemma}
	\label{lem:MWSelGadgetTSStructure}
	Let $L=L(s,r)$ be a~selection gadget in~$G'$ and suppose ${S \subseteq V(G')}$ is a~target set that contains exactly $s$ vertices from $V(L)$.
  Then, the~following holds.
	\begin{enumerate}
		\item The~vertices of $V(L) \cap S$ are in the~selection part of the~gadget $L$.\label{itm:SelGadgetTSIsInIface}
		\item Suppose that $v$ is a vertex in the selection part of $L$ but not in $S$. The~vertex $v$ cannot be activated before all neighboring vertices of $L$ adjacent to $v$ are activated.\label{itm:SelGadgetNeighborstFirst}
	\end{enumerate}
\end{lemma}
\begin{proof}
	The~only vertices with deficit $s$ are in the~$L$-guard part, others have their deficit strictly higher.
  Suppose $S$ contains less than $s$ vertices from the neighborhood of $L$-guard (i.e., from the selection part of $L$ -- $L$-pos or $L$-neg).
  To activate $L$-guard some other vertices in the selection part have to be activated.
  Suppose further that the activation process manages to activate the rest of the graph, i.e., all~$r$ neighbors of the gadget $L$ are activated.
  Even in this case, the vertices in the selection part have at most $s + r$ active neighbors and their threshold is $3s + r$.
  Thus, $L$-doubling has to be activated.
  However, vertices in $L$-doubling have threshold $2s$ and they have at most $s$ active neighbors.
  To activate $L$-doubling, the $L$-end type has to be activated.
  Again, vertices $L$-end have threshold strictly higher than $s$ and they have at most $s$ active neighbors.
  Therefore, the activation process would stop and the only way how to proceed is that~$S$ contains $s$ vertices in the~selection part of~$L$ ($L$-pos and $L$-neg).

	For the~second part, we can assume $v \in V(\mbox{$L$-pos}) \setminus S$, since the~gadget is symmetric.
  From Part~\ref{itm:SelGadgetTSIsInIface} we already know that the~set $S$ contains only vertices from the selection part.
  Thus, $L$-guard is activated in the first round.
  To activate $v$, other $r$ vertices has to be activated.
  Since $S$ does not intersect $L$-doubling and $L$-end, those types can be activated after $L$-pos and $L$-neg.
  Thus, the neighboring vertices of $L$ connected to $v$ has to be activated before $v$.
\end{proof}

\subsubsection{Threshold Decrease Gadget $D(d,r)$}
The~purpose of a threshold decrease gadget $D(d,r)$ is to ensure there are $d$ active vertices in its interface part.
This can be also viewed as effectively decreasing the~threshold of neighboring vertices by $d/2$, since if it works as described, then it adds~$d$ neighbors to the neighborhood of the gadget all of which are activated.
The~parameter $r$ again determines the~number of neighboring vertices.

The~gadget consists of the~following three types (see Figure~\ref{fig:thrDecGadgetMajTSS} for an overview):
\begin{itemize}
	\item The~independent set of size $d$ called \emph{$D$-interface}.
	\item The~independent set of size $2d + r + 1$ called \emph{$D$-middle}.
	\item The~independent set of size $d$ called \emph{$D$-end}.
\end{itemize}
The type~$D$-middle is connected to $D$-interface and $D$-end.
\begin{figure}[bt]
	\begin{center}
	\usetikzlibrary{positioning,fit,calc}

\begin{tikzpicture}[node distance=2.5cm]
  \tikzstyle{bag} = [circle, draw, minimum width=1.3cm]
  \tikzstyle{edge} = [black, ultra thick]
  \tikzstyle{dummy} = [circle, dotted, draw, minimum width=1.3cm]
  \tikzstyle{fitting} = [draw, dashed, rounded corners, black!40, thick, inner sep=18pt]

  \node[bag, label={90:$D$-{\rm interface}}] (DDI) {$d$};
  \node[bag, left of=DDI, label={90:$D$-{\rm middle}}] (DD) {$2d+r+1$};
  \node[bag, left of=DD, label={90:$D$-{\rm end}}] (DDD) {$d$};

  \node[dummy, right of=DDI] (DDO) {$r$};

  \draw[edge] (DDD) -- (DD) -- (DDI) -- (DDO);

  \node[fitting, fit=(DDI), label={90:interface}] {};
\end{tikzpicture}
	\end{center}
	\caption{An~overview of the threshold decrease gadget for \MAJTSS}
	\label{fig:thrDecGadgetMajTSS}
\end{figure}

\begin{lemma}
	\label{lem:MWThrDecGadgetMin}
	Let $D = D(d,r)$ be a~threshold decrease gadget in the~graph $G'$. If $S \subseteq V(G)$ contains less than $d$ vertices in $V(D)$ then $S$ is not a~target set.
  Moreover, $S \cap V(D)$ is a subset of vertices in $D$-interface and $D$-end.
\end{lemma}
\begin{proof}
	Let $S_D$ denote $S \cap V(D)$.
	Again we can assume that vertices in $\overline{V_D}$ are in $S$.
	The~deficit of vertices in the $D$-end part is ${(2d+r+1)/2 > d}$, the~deficit of vertices in the $D$-middle part is $d$, and the~deficit of vertices in $D$-interface part is $(2d+2r+1)/2 > d$.
  Therefore, by Lemma~\ref{lem:deficitLowerBound} the set~$S$ is not a target set.
  The moreover part of the lemma follows immediately, since the only part with deficit exactly~$d$ is the $D$-middle part and therefore~$S$ must contain~$d$ vertices in its neighborhood.
\end{proof}

Note that we can choose how we distribute $d$ vertices in $S$ between $D$-end and $D$-interface.
The next lemma states that for determining the existence of a target set, it is enough to consider sets $S$ that have $d$ vertices in $D$-interface.

\begin{lemma}
	\label{lem:canonicalizationOfTS}
	Let $D = D(d,r)$ be a threshold decrease gadget in the~graph $G'$.
  If $S$ is a target set, then the set $S' = \big(S \setminus V(D)\big) \cup V(\mbox{$D$-{\rm interface}})$ is a target set and $|S'| \leq |S|$.
\end{lemma}
\begin{proof}
We first verify that $|S'| \leq |S|$.
The sets $S$ and $S'$ coincide on $V(G') \setminus V(D)$.
From Lemma~\ref{lem:MWThrDecGadgetMin} we have that $|S \cap V(D)| \geq d$.
From the definition of $S'$ we see that $|S' \cap V(D)| = d$.
Putting these three observations together, we obtain that $|S'| \leq |S|$.

Now observe that if we start the activation process with the set $S'$ the vertex set of the gadget~$D$ activates after two rounds.
In particular, we have $S \subseteq S'_2$.
Consequently, the activation process starting from the set $S'$ must activate all the vertices activated in the activation process starting from the set $S$.
Since $S$ was a target set, so is~$S'$.
\end{proof}

We set the budget $b$ in such a~way that it exactly covers all selection gadget and all threshold decrease gadgets.
This means that if any element of $S$ is outside of the~selection part of the selection gadgets or outside of the $D$-interface or the $D$-end part of the threshold decrease gadgets, then $S$ is not a~target set. By Lemma~\ref{lem:canonicalizationOfTS}, we know that we can ignore $S$ intersecting any $D$-end part. This motivates the~following definition.

\begin{definition}
	\label{def:hopefulSet}
	We say that a~set $S \subseteq V(G')$ is \emph{hopeful} if all its elements lie in the~selection parts of selection gadgets and in the~interface parts of threshold decrease gadgets.
\end{definition}

When the~budget is set as mentioned, it follows from Lemma~\ref{lem:MWSelGadgetMin}, Lemma~\ref{lem:MWThrDecGadgetMin}, and Lemma~\ref{lem:canonicalizationOfTS} that if we want to decide whether there exists a~target set, it is sufficient to consider only hopeful sets.

\begin{definition}
	Let $S \subseteq V(G')$ be a hopeful set and $L = L(s,r)$ be a selection gadget in $G'$. If there are $i$ elements of $S$ in $L$-pos and $s-i$ elements of $S$ in $L$-neg, we say that $S$ \emph{represents the number $i$ in $L$}.
\end{definition}

\subsubsection{Check Gadget $C(Z,s)$}
The parameter $s$ is a positive integer and $Z \subseteq \{0,\dots,s\}$.
The~gadget has two interface parts $C$-pos and $C$-neg that will be connected to the pos and neg parts of selection gadgets.
The~gadget is constructed in such a~way that it is activated if and only if the~sum of the~numbers represented by connected selection gadgets is in $Z$.

We start with few auxiliary definitions.
An~\emph{$i$-clique group} is the~graph ${(2(s-i)+1)\times K_{2i+1}}$ (i.e., disjoint union of $2(s-i)+1$ cliques, each of size $2i+1$).
The~\emph{weight} of an~$i$-clique group is the~number of vertices of the~group and is denoted by $w_s(i) = (2(s-i)+1) \cdot (2i+1)$.

The~graph $C_Z$ is defined as a~disjoint union of $z$-clique groups for each $z\in Z$.
The~graph $\ol{C}_Z$ is defined as a~disjoint union of $(s-z)$-clique groups for each $z \in Z$.
The $z$-clique group in $C_Z$ and the $(s-z)$-clique group in $\ol{C}_Z$ are referred to as \emph{complementary clique groups}.
Note that the clique group and its complementary clique group have the same weight.

The~\emph{weight of the set $Z$} is defined as
\[
  w_s(Z) = \sum_{z \in Z} w_s(z) \,.
\]

The~gadget $C(Z,s)$ consists of the~following parts (see Figure~\ref{fig:MWCheckGadget} for overview):
\begin{itemize}
	\item Interface parts $C$-pos and $C$-neg that are formed by $C_Z$ and $\ol{C}_Z$, respectively.
	\item Independent set $C$-guard of size $s$ that is connected to the parts $C$-pos and $C$-neg.
	\item A single vertex connected to $C$-guard.
	\item A threshold decrease gadget $D_C = D\bigl(2s, 2w_s(Z)\bigr)$ with its interface part connected to $C$-pos and $C$-neg.
\end{itemize}
\begin{figure}[bt]
	\usetikzlibrary{positioning,fit,calc,shapes}

\begin{tikzpicture}[node distance=2.2cm,scale=0.80,font=\footnotesize]
  \tikzstyle{bag} = [circle, draw, minimum width=.7cm]
  \tikzstyle{edge} = [black, ultra thick]
  \tikzstyle{dummy} = [circle, dotted, draw, minimum width=1.3cm]
  \tikzstyle{fitting} = [draw, dashed, rounded corners, black!40, thick, inner sep=4pt,text=black,minimum width=2.5cm]
  \tikzstyle{superbag}=[draw, ellipse, inner sep=0pt]
  \tikzstyle{dot}=[black, circle, fill, minimum width=.1cm, inner sep=2pt]

  \newcommand{\trojiceTecek}[1]{
    \node[dot] at ($(#1.north)!.3!(#1.south)$) {};
    \node[dot] at ($ ($(#1.west)!0.66!(#1.south west)$)!.2!($(#1.east)!0.33!(#1.north east)$) $) {};
    \node[dot] at ($ ($(#1.east)!0.66!(#1.south east)$)!.2!($(#1.west)!0.33!(#1.north west)$) $) {};
  }

  \newcommand{\dvojiceTecek}[1]{
    \node[dot] at ($(#1.north)!.3!(#1.south)$) {};
    \node[dot] at ($(#1.north)!.7!(#1.south)$) {};
  }

  %%% upper super bag
  \begin{scope}[node distance=1.6cm]
    \node[bag] at (0,4.4) (UPX) {};  %% upper
    \node[bag, below of=UPX] (UPY) {}; %% lower
    \node[yshift=3pt] at ($(UPX)!0.5!(UPY)$) (UPXY) {$\vdots$};
    \node[superbag, fit=(UPX) (UPXY) (UPY)] (UP) {};
    \trojiceTecek{UPX};
    \dvojiceTecek{UPY};
  \end{scope}

  %%% lower super bag
  \begin{scope}[node distance=1.6cm]
    \node[bag] (DWX) {};  %% upper
    \node[bag, below of=DWX] (DWY) {}; %% lower
    \node[yshift=3pt] at ($(DWX)!0.5!(DWY)$) (DWXY) {$\vdots$};
    \node[superbag, fit=(DWX) (DWXY) (DWY)] (DOWN){};
    \trojiceTecek{DWY};
    \dvojiceTecek{DWX};
  \end{scope}

  \node[fitting, fit=(UP) (DOWN),label={90:Interface}] {};

  %%% dummy bag
  \node[dummy, left of=UPXY,yshift=-3pt] (dummyUP) {$s$};
  \node[dummy, left of=DWXY,yshift=-3pt] (dummyDW) {$s$};

  %%% checking nodes
  \node[bag] at ($(UP)!.5!(DOWN) + (2.3cm,0)$) (CHECK) {$s$};
  \node[bag, minimum width=0cm] at ($(CHECK.east) + (0.8cm,0)$) (CHECK1) {$1$};

  %%% threshold decrease
  \node[bag] at  ($(CHECK.east) + (2.5cm,0)$) (DDI) {$2s$};
  \node[bag, right of=DDI, text width=1.3cm,align=center] (DD) {$4s+2w_s(Z)+1$};
  \node[bag, right of=DD] (DDD) {$2s$};

  \draw[edge] (DDD) -- (DD) -- (DDI);
  \node[fitting, fit=(DDD) (DD) (DDI), label={90:Threshold decrease gadget $D(2s,2w_s(Z))$}] {};

  %%% edges
  \draw[edge] (dummyUP) -- (UP);
  \draw[edge] (dummyDW) -- (DOWN);
  \draw[edge] (CHECK) -- (CHECK1);
  \draw[edge] (UP) -- (CHECK);
  \draw[edge] (DOWN) -- (CHECK);
  \draw[edge] (UP) -- (DDI);
  \draw[edge] (DOWN) -- (DDI);
\end{tikzpicture}
	\caption{An~overview of Check gadget $C(Z,s)$ for \MAJTSS}
	\label{fig:MWCheckGadget}
\end{figure}

\begin{lemma}
	\label{lem:MWCheckGadgetCorrectness}
	Let $C = C(Z,s)$ be a check gadget in the graph $G'$ and let ${L_1 = L(s_1,r_1)}, \ldots,$ ${L_h = L(s_h,r_h)}$ be selection gadgets that are connected to $C$ in such a way that each $L_i$-pos is connected to $C$-pos and each $L_i$-neg is connected to $C$-neg.
	We additionally require that $\sum s_i = s$ and no other gadget is connected to~$C$.
	Furthermore, let $S$ be a hopeful set that represents numbers $\ell_1, \ldots, \ell_h$ in gadgets $L_1, \ldots L_h$, respectively.
	Then, the gadget $C$ is activated if and only if $\ell_1 + \cdots + \ell_h \in Z$.
\end{lemma}
\begin{proof}
	First notice that due to Lemma~\ref{lem:MWSelGadgetTSStructure} Part~\ref{itm:SelGadgetNeighborstFirst}, no vertex in $L_i$ can be activated before $C$-pos and $C$-neg are fully activated.
	Let us set $\ell = \ell_1 + \cdots + \ell_h$.
	\begin{claimproof}
		\label{clm:cliqueGroupActivation}
		The $i$-clique group $K$ in $C_Z$ is activated in the first round if and only if $i \leq \ell$.
		The $(s-i)$-clique group $\ol{K}$ in $\ol{C}_Z$ is activated in the first round if and only if $s - i \leq s - \ell$.
	\end{claimproof}
	The degree of vertices in the $i$-clique group $K$ is $4s + 2i$; there are $s$ neighbors in the adjacent selection gadgets, $s$ neighbors in the $C$-guard part, $2s$ neighbors in the adjacent threshold decrease gadget, and $2i$ neighbors within the clique itself.
	Therefore, the threshold of vertices in the $i$-clique group is $2s + i$.
	We know that $2s$ vertices in the neighborhood of $K$ are already activated due to the threshold decrease gadget $D_C$.
	Thus, at least $i$ other vertices need to be activated for the $i$-clique to be activated.
	Since $S$ is hopeful, these vertices can be only in the pos parts of the adjacent selection gadgets.
	From the assumption of the lemma, we know that there are $\ell$ such vertices, which gives us that the clique group activates if and only if $\ell \geq i$.

	To prove the second part, it is enough to observe that exactly $s - \ell$ vertices are activated in the neg parts of the selection gadgets $L_1 \dots, L_h$ because $S$ is hopeful.
  The rest follows by a similar argument.
	\begin{claimproof}
		\label{clm:complementaryActivation}
		If a clique group $K$ is not activated in the first round, then its complementary clique group $\ol{K}$ is activated in the first round.
	\end{claimproof}
	Suppose that $K$ is an $i$-clique group.
	This means that $\ol{K}$ is an $(s-i)$-clique group.
	If $K$ is not activated in the first round, then we have $i > \ell$ by Claim~\ref{clm:cliqueGroupActivation}.
	But then $s-i < s- \ell$ and again by Claim~\ref{clm:cliqueGroupActivation} the clique group $\ol{K}$ is activated in the first round.

  Observe that the complementary clique groups $K$ and $\ol{K}$ have the same number of vertices and at least one of them is activated after the first round by Claim~\ref{clm:complementaryActivation}.
	\begin{claimproof}
		\label{clm:bothActivate}
		The $i$-clique group $K$ in $C_Z$ and its complementary clique group $\ol{K}$ are activated in the first round if and only if $\ell = i$ and therefore $\ell \in Z$.
	\end{claimproof}
	From Claim~\ref{clm:cliqueGroupActivation} we see that both $K$ and $\ol{K}$ are activated in the first round if and only if $i \leq \ell$ and $s - i \leq s - \ell$, which together gives $i = \ell$.
	We put an $i$-clique group into $C_Z$ if and only if $i \in Z$, which finishes the proof of the claim.

	\begin{claimproof}
		If $\ell \notin Z$, the activation process in $C$ stops after the first round.
	\end{claimproof}
	If $\ell \notin Z$, by Claim~\ref{clm:bothActivate} we see that there is no pair of complementary clique groups $K$, $\ol{K}$ such that both of them are activated in the first round.
	But then $C$-guard cannot be activated, as less than half of its neighbors are activated.
  Note that the interface vertices of the gadget $C$ have to be activated in order to activate the rest of the vertices in adjacent selection gadgets (by Lemma~\ref{lem:MWSelGadgetTSStructure}).

	\begin{claimproof}
		If $\ell \in Z$, the whole gadget~$C$ is activated in three rounds.
	\end{claimproof}
	We can only consider vertices outside of threshold decrease gadgets; it is straightforward to check that if $S$ is hopeful, all threshold decrease gadgets are activated in two rounds.

	If $\ell \in Z$, then by Claim~\ref{clm:bothActivate} a pair of complementary clique groups $K$, $\ol{K}$ are both activated in the first round.
	Thus, a strict majority of vertices in $C$-pos and $C$-neg are activated in the first round.
	This causes the $C$-guard to be activated in the second round.
	In the third round, the single vertex connected to $C$-guard is activated.
	Moreover, all remaining clique groups are activated, as they have at least $3s$ activated neighbors in total ($2s$ from the threshold decrease gadget, $s$ from $C$-guard); see Claim~\ref{clm:cliqueGroupActivation}.
	This finishes the proof of Lemma~\ref{lem:MWCheckGadgetCorrectness}.
\end{proof}

\subsubsection{Multiple Gadget $M(s,q,r)$}
The last gadget used is a multiple gadget.
It is a variant of selection gadget; similarly to the selection gadget, every target set represents an integer between $0$ and $qs$.
However, the multiple gadget additionally ensures that the number is a multiple of $q$.
The last parameter $r$ determines the number of neighboring vertices, as usual.
The gadget $M$ consists of a selection gadget $L = L(qs, r+w_{qs}(Q_{q,s}))$ and a check gadget $C = C(Q_{q,s},qs)$, where $Q_{q,s} = \{0,q,2q,\ldots,sq\}$.
Observe that the multiple gadget contains a threshold decrease gadget $D(2qs, w_{qs}(Q_{q,s}))$ as a part of the check gadget $C(Q_{q,s},qs)$.
The type $L$-pos is connected to $C$-pos and $L$-neg is connected to $C$-neg.
The interface of $M$ is the interface of the gadget $L$, i.e., $L$-pos and $L$-neg.

\begin{lemma}
	\label{lem:MWMultipleGadgetMinimal}
	Let $M = M(s,q,r)$ be a multiple gadget in $G'$. If $S \subseteq V(G)$ satisfies $|S \cap V(M)| < 3qs$, then $S$ is not a target set.
\end{lemma}
\begin{proof}
	If $|S \cap V(M)| < 3qs$, then at least one of the following necessarily happens.
	\begin{itemize}
		\item
      The selection gadget in $M$ has less than $qs$ vertices from $S$.
      Then by Lemma~\ref{lem:MWSelGadgetMin} we have that $S$ is not a target set.
		\item
      The threshold decrease gadget in $M$ has less than $2qs$ vertices from $S$.
      By Lemma~\ref{lem:MWThrDecGadgetMin} again $S$ is not a target set.
	\end{itemize}
\end{proof}

\begin{lemma}
	\label{lem:MWMultipleGadgetTSStruct}
	Let $M = M(s,q,r)$ be a multiple gadget in $G'$ with the selection gadget $L$ and let $S \subseteq V(G')$ be a target set with $|S \cap V(M)| = 3qs$. Then, the following holds.
	\begin{enumerate}
		\item\label{itm:MWMultSelPart}The $qs$ vertices of $S$ are in the $L$-pos and $L$-neg types.
		\item\label{itm:MWMultActivationOrder} Suppose $v$ is a vertex in $M$-interface but not in $S$. The vertex $v$ cannot be activated before all neighboring vertices of $M$ adjacent to $v$ are activated.
		\item\label{itm:MWMultSelectsMultiple}The number of vertices of $|S \cap V(\mbox{$L$-pos})|$ is a multiple of $q$.
	\end{enumerate}
\end{lemma}
\begin{proof}
	For a proof of Part~\ref{itm:MWMultSelPart} first observe that $qs$ vertices of $S$ must be in the selection gadget in $M$.
	There cannot be more by the fact that we have $3qs$ vertices in total and by Lemma~\ref{lem:MWThrDecGadgetMin}; on the other hand, there cannot be less by Lemma~\ref{lem:MWSelGadgetMin}.
	Now these $qs$ vertices must be in $L$-pos and $L$-neg by Lemma~\ref{lem:MWSelGadgetTSStructure} Part~\ref{itm:SelGadgetTSIsInIface}.

	Part~\ref{itm:MWMultActivationOrder} follows immediately from Lemma~\ref{lem:MWSelGadgetTSStructure} Part~\ref{itm:SelGadgetNeighborstFirst}.

	Finally for Part~\ref{itm:MWMultSelectsMultiple}, if $|S \cap V(\mbox{$L$-pos})|$ is not a multiple of $q$, then the check gadget of $M$ cannot be activate by
	Lemma~\ref{lem:MWCheckGadgetCorrectness} and hence $S$ is not a target set.
\end{proof}

\subsubsection{Finishing the Reduction}
We finally have all building blocks for the reduction.
The graph $G'$ is constructed as follows (the numbers $r_c$, $r_{cd}$, $r_{c:cd}$ will be computed later).

\begin{itemize}
	\item For each color class $V_c$ we have a selection gadget $L_c = L(n,r_c)$.
	\item For each edge class $E_{cd}$ we have a multiple gadget $M_{cd} = M(m_{cd},2n,r_{cd})$, where $m_{cd} = |E_{cd}|$.
	\item For each ordered pair $(c,d)$ of vertices of $H$ such that $\{c,d\} \in E(H)$, we have a check gadget $I_{c:cd} = C(Z_{c:cd}, r_{c:cd})$. The $I_{c:cd}$-pos part of  is connected to $L_c$-pos and $M_{cd}$-pos. The type $I_{c:cd}$-neg is connected analogously.
\end{itemize}

\begin{figure}[bt]
	\begin{center}
		\usetikzlibrary{positioning,fit,calc}

\begin{tikzpicture}[node distance=1.7cm,scale=0.90]
  \tikzstyle{bag} = [circle, draw, minimum width=1cm]
  \tikzstyle{edge} = [black, ultra thick]
  \tikzstyle{otherEdge} = [black, ultra thick, dashed]
  \tikzstyle{dummy} = []
  \tikzstyle{fitting} = [draw, dashed, rounded corners, black!40, thick, inner sep=2pt]

  %%% Vertices
  \node[bag] (VertIntUP) {$L_c$-pos};
  \node[bag, below of=VertIntUP] (VertIntDOWN) {$L_c$-neg};
  \node[dummy] at ($(VertIntUP.west) - (2cm, 0)$) (VertDummy) {$V_c$};
  \node[fitting, fit=(VertIntUP) (VertIntDOWN) (VertDummy), label={90:$L_c = L(n,r_c)$}] {};

  %%% Incedence
  \node[bag] at ($(VertIntUP.east) + (1.5cm,0)$) (IncIntUP) {$I_{c:cd}$-pos};
  \node[bag, below of=IncIntUP] (IncIntDOWN) {$I_{c:cd}$-neg};
  \node[dummy] at ($(IncIntDOWN.east) + (2cm,0)$) (IncDummy) {};
  \node[fitting, fit=(IncIntUP) (IncIntDOWN) (IncDummy), label={90:$I_{c:cd} = C(Z_{c:cd},n+2nm)$}] {};

  \node[dummy, below of=IncDummy,xshift=-.5cm] (IncDummy2) {};
  \node[dummy, below of=IncDummy2] (IncDummy3) {};

  \node[dummy, below of=IncIntDOWN] (IncDummyVert1) {};
  \begin{scope}[node distance=.5cm]
    \node[dummy, below of=IncDummyVert1] (IncDummyVert2) {};
    \node[dummy, below of=IncDummyVert2] (IncDummyVert3) {};
    \node[dummy, below of=IncDummyVert3] (IncDummyVert4) {};
  \end{scope}

  %%% Edges
  \node[bag] at ($(IncIntUP.east) + (4cm,0)$) (EdgeIntUP) {$M_{cd}$-pos};
  \node[bag, below of=EdgeIntUP] (EdgeIntDOWN) {$M_{cd}$-neg};
  \node[dummy] at ($(EdgeIntUP.east) + (2.5cm, 0)$) (EdgeDummy) {$E_{cd}$};
  \node[fitting, fit=(EdgeIntUP) (EdgeIntDOWN) (EdgeDummy), label={90:$M_{cd} = M(m,2n,r_{cd})$}] {};

  %%% Graph Edges
  \draw[edge] (VertIntUP) -- (IncIntUP);
  \draw[edge] (VertIntDOWN) -- (IncIntDOWN);

  \draw[edge] (IncIntUP) -- (EdgeIntUP);
  \draw[edge] (IncIntDOWN) -- (EdgeIntDOWN);

  \draw[otherEdge] (IncDummy2) -- (EdgeIntUP);
  \draw[otherEdge] (IncDummy3) -- (EdgeIntDOWN);
  \draw[otherEdge] (IncDummyVert1) -- (VertIntUP);
  \draw[otherEdge] (IncDummyVert2) -- (VertIntUP);
  \draw[otherEdge] (IncDummyVert3) -- (VertIntDOWN);
  \draw[otherEdge] (IncDummyVert4) -- (VertIntDOWN);

\end{tikzpicture}
	\end{center}
	\caption{Overview of the reduction.}
	\label{fig:bigPicture}
\end{figure}

We numerate vertices in each class $V_c$ by numbers from $0$ to $n$ and edges in each class $E_{cd}$ by numbers from $0$ to $m_{cd}$.
Then if $V_c = \{v_0, \ldots, v_n\}$ and $E_{cd} = \{e_0, \ldots, e_{m_{cd}}\}$, we set
\[
  Z_{c:cd} = \setof{ i + 2nj }{ v_i \in e_j } \,.
\]

It remains to compute the numbers $r_c$, $r_{cd}$, and $r_{c:cd}$.
The number $r_{c:cd}$ is $n + 2nm_{cd}$, as $I_{c:cd}$-pos is adjacent to $L_c$-pos and $M_{cd}$-pos.

The gadget $M_{cd}$ is adjacent to $I_{c:cd}$ and $I_{d:cd}$.
To determine $r_{cd}$, we need to count the number of vertices of the pos types in both incident gadgets.
This can be expressed by the weight function giving
\[
  r_{cd} = w_{r_{c:cd}}(Z_{c:cd}) + w_{r_{d:cd}}(Z_{d:cd}).
\]
Finally, the gadget $L_c$ is adjacent to $I_{c:cd}$ for every $d \in N_H(c)$.
analogously, we get
\[
  r_c = \sum_{d \in N_H(c)} w_{r_{c:cd}}(Z_{c:cd}).
\]
The budget $b$ is set as
\[
  b = nk + 6n \cdot \bar{m} + 2n|E(H)| + 4n \cdot \bar{m} \,,
\]
where $\bar{m}$ is again $\sum_{\{c,d\} \in E(H)}m_{cd}$.
The first term $nk$ is for $k$ selection gadgets $L(n,r_c)$ corresponding to the color classes.
The second term accounts the multiple gadgets $M(m_{cd},2n,r_cd)$ corresponding to the edge sets $E_{cd}$; each of them requires at least $6nm_{cd}$ vertices of a target set by Lemma~\ref{lem:MWMultipleGadgetMinimal}.
The last two terms are for threshold decrease gadgets $D\bigl(2(2m_{cd}n+n),w_{2m_{cd}n+n}(Z_{c:cd})\bigr)$ inside each check gadget $I_{c:cd} = C(Z_{c:cd},2m_{cd}n+n)$.

We verify that $G'$ has size polynomial in $|V(G)|$.
Observe that the size of each gadget is polynomial in its parameters, and the parameters are polynomial with respect to $|V(G)|$.
Since we have $\BigO(k)$ gadgets, the whole graph $G'$ has size polynomial in $|V(G)|$, since $k \le |V(G)|$.
Moreover, the graph $G'$ can be clearly constructed in polynomial time.

We turn our attention to the restricted modular width of $G'$.
Consider a graph $\hat{G}$ obtained by replacing each pos and neg part of each check gadget in $G'$ by an independent set of the same size.
It is easy to verify that $\hat{G}$ is a graph of neighborhood diversity $\BigO(k)$; every gadget in $\hat{G}$ has neighborhood diversity bounded by a constant.
As each pos or neg part of the check gadget in $G'$ is a cluster graph, we can obtain $G'$ from the type graph $T_{\hat{G}}$ by an appropriate substitution operation.
This gives us an algebraic expression proving the restricted modular width of $G'$ is $\BigO(k)$.

\begin{lemma}
  The graph $G$ contains a colored copy of $H$ if and only if $G'$ has a target set of size $b$.
\end{lemma}
\begin{proof}
  First, suppose that $G$ contains a colored copy of $H$.
  Denote by $\ell_c$ the number of the selected vertex in $V_c$ in the numeration we fixed before.
  Similarly denote by $\ell_{cd}$ the number of edge that connects the selected vertices in $V_c$ and $V_d$.

  The target set is constructed as follows.
  \begin{enumerate}
  	\item Put $\ell_c$ vertices of $L_c$-pos and $n - \ell_c$ vertices of $L_c$-neg into $S$.
  	\item Put $2n\ell_{cd}$ vertices of $M_{cd}$-pos and $2m_{cd}n - 2n\ell_{cd}$ vertices of $M_{cd}$ into $S$.
  	\item Put all vertices of the interface part of every threshold decrease gadget into $S$.
  \end{enumerate}
  By the choice of $b$, we see that $|S| = b$.
  We claim that $S$ is a target set.

  By Lemma~\ref{lem:MWCheckGadgetCorrectness}, the check gadget $C(Q_{2n,m_{cd}},2nm_{cd})$ in $M_{cd}$ is activated in three rounds as $\ell_{cd}$ is a multiple of $2n$ and hence is in $Q_{2n,m_{cd}}$.
  Similarly, as the vertex $v \in V_c$ corresponding to $\ell_c$ is incident to the edge $e \in E_{cd}$ corresponding to $\ell_{cd}$, the gadget $I_{c:cd}$ is activated in three rounds again by Lemma~\ref{lem:MWCheckGadgetCorrectness}, since the number $\ell_c + 2n\ell_{cd}$ is in $Z_{c:cd}$.

  We now analyze the gadget $L_c = L(n,r_c)$.
  The part $L_c$-guard is activated in the first round.
  In the third round, all neighboring check gadgets are activated.
  Since those gadgets together with $L_c$-guard constitute exactly the half of the neighbors of every vertex in $L_c$-pos or $L_c$-neg, these parts are activated in the round four.
  It is now easy to see that the remaining two types ({$L_c$-doubling} and {$L_c$-end}) are activated in the rounds five and six.

  The argument for the gadget $M_{cd} = M(m_{cd},2n, r_{cd})$ is analogous -- in the third round, all check gadgets (including the one inside $M_{cd}$) are activated and again this is enough to activate the $M_{cd}$-pos and $M_{cd}$-neg parts in the round four.
  The rest of the gadget again is activated in the round six.
  This proves that $S$ is a target set.

  For the opposite direction, suppose that we have a target set $S$ in $G'$ with $|S| \leq b$.
  We know that without loss of generality we may assume that $S$ is hopeful.
  Let $\ell_c$ be the number represented by $S$ in $V_c$ and $2n \ell_{cd}$ the number represented by $S$ in $M_{cd}$.
  We know the number represented by $S$ in $M_{cd}$ is a multiple of $2n$ by Lemma~\ref{lem:MWMultipleGadgetTSStruct}.
  Now we pick $\ell_c$-th vertex from each color class $V_c$.
  We claim that those vertices form a solution to the original instance of CSI.

  Since $S$ is a target set, we obtain from Lemma~\ref{lem:MWCheckGadgetCorrectness} that
  $\ell_c + 2n\ell_{cd} \in Z_{c:cd}$.
  Note that a number $z \in Z_{c:cd}$ is of the form $z = i + 2nj$ and the number $i$ and $j$ are uniquely determined by $z$ alone; simply take $i = z \mathbin{\mathrm{mod}} 2n$ and $j = z \mathbin{\mathrm{div}} 2n$.
  This means that number $\ell_c + 2n\ell_{cd}$ is in $Z$ if and only if a vertex $v \in V_c$ corresponding to the number $\ell_c$ is incident to an edge $e \in E_{cd}$ corresponding to the number $\ell_{cd}$.
  Thus, we have proved that the set of vertex obtained from $S$ indeed form a solution to~$(G, H, \psi)$.
\end{proof}
This finishes the proof of Theorem~\ref{thm:MAJTSSisWwrtMW}.

\subsection{Twin Cover}
As usual, we enumerate vertices in color classes $V_c$ and edges in sets $E_{cd}$ and design vertex, edge, and incidence gadgets based on this enumeration.
This time all the gadgets are cliques and we encode the input instance into sizes of parts of these cliques with different thresholds.
It is worth noting that in each such clique the number of different vertex thresholds is at most eleven.
The resulting instance $(G', f, b)$ of \TSS is equivalent to the instance $(G,H)$ of CSI and $\tc(G') = \BigO(k)$.

\subsubsection{Construction}
Let $c,d\in V(H)$.
As we assume $|V_c| = |V_d|$, we denote this common value as $n - 1$ for technical reason.
Thus, we can denote the vertices in $V_c$ as $v^c_1,\dots,v^c_{n-1}$.
As we mentioned, we have no control on the number of edges in $E_{cd}$ -- it is upper bounded by $(n-1)^2$.

We arbitrarily orient edges of $H$.
For a vertex $c$ by $\indeg(c)$ we denote the \emph{in-degree} of $c$, that is, number of ingoing arcs towards $c$ in this orientation; analogously by $\outdeg(c)$ is the \emph{out-degree} of $c$.
Since $H$ is 3-regular, $\deg(c) = \indeg(c) + \outdeg(c) = 3$ for every vertex $c \in V(H)$.
We divide the vertices of $H$ into 3 groups according to their in-degree and out-degree.
For $s \in \{+,-\}$, let $V^s \subseteq V(H)$ be a set of vertices $c \in V(H)$ such that $\deg^s(c) = 3$ and $V^r = V(H) \setminus (V^+ \cup V^-)$.
%In what follows please refer to Figures~\ref{fig:TCTSSBigPicture} and~\ref{fig:TCTSSOverview}.

\subparagraph{Vertex Gadget}
Let $c \in V(H)$ and let $v = v^c_i$ be the $i$-th vertex in $V_c$.
We first introduce few constants that will later allow us to verify the proof of the correctness of our reduction.
Let $\gamma_s = 21$, $\gamma_\ell = 10$, and $\Gamma = \gamma_s + \gamma_\ell + 1$ ($\gamma_s$ will denote the size of a vertex set which we will call a selector, $\gamma_\ell$ will denote the size of a vertex set which we will a call class validator and denote $L$).
For the vertex $v$ we create a clique $C_v$.
The size of $C_v$ and thresholds of the vertices in $C_v$ depend on $\outdeg(c)$ and $\indeg(c)$.
The main idea is that in the activation process we first check the incidence of vertices and outgoing arcs and then we check the incidence of ingoing arcs.
Thus, for example if $c \in V^+$ (i.e., $\outdeg(c) = 0$), then there is no need to check incidence between selected $v \in V_c$ and outgoing edges.
The vertices in $C_v$ are divided into several groups.
The sizes of the groups and thresholds of vertices in the group are in the following table.

\begin{center}
\begin{tabular}{cc|c||c|c||c|c|}
 \cline{2-7}
 \multicolumn{1}{c}{} &
 \multicolumn{2}{||c||}{$c \in V^r$} &
 \multicolumn{2}{c||}{$c \in V^-$} &
 \multicolumn{2}{c|}{$c \in V^+$} \\
 \hline
 \multicolumn{1}{|c}{Name} & \multicolumn{1}{||c|}{Size} & Threshold & Size & Threshold & Size & Threshold\\
 \hline
 \multicolumn{1}{|c}{selector} & \multicolumn{1}{||c|}{$\gamma_s$} & $\deg - \gamma_s + 1$ & $\gamma_s$ & $\deg - \gamma_s + 1$ & $\gamma_s$ & $\deg - \gamma_s + 1$ \\
 \hline
 \multicolumn{1}{|c}{starter} & \multicolumn{1}{||c|}{$1$} & $\gamma_s$ & {$1$} & $\gamma_s$ & {$1$} & $\gamma_s$ \\
 \hline
 \multicolumn{1}{|c}{1.} & \multicolumn{1}{||c|}{$i$} & $\Gamma$ & $i$ & $\Gamma$  & $2n + i$ & $\Gamma$ \\
 \hline
 \multicolumn{1}{|c}{2.} & \multicolumn{1}{||c|}{$n - i$} & $i + \Gamma + \outdeg(c)$ & $n - i$ & $i + \Gamma + 3$ & &\\
 \hline
 \multicolumn{1}{|c}{3.} & \multicolumn{1}{||c|}{$n - i$} & $n + \Gamma + \outdeg(c) + 1$ & $n - i$ & $n + \Gamma + 4$ & &\\
 \hline
 \multicolumn{1}{|c}{4.} & \multicolumn{1}{||c|}{$i$} & $2n - i + \Gamma + 2\outdeg(c) + 1$ & $i$ & $2n - i + \Gamma + 7$ & &\\
 \hline
 \multicolumn{1}{|c}{5.} & \multicolumn{1}{||c|}{$i$} & $2n + \Gamma + 2\outdeg(c) + 2$ & $n^3 + 2n$ & $2n + \Gamma + 8$ & &\\
 \hline
 \multicolumn{1}{|c}{6.} & \multicolumn{1}{||c|}{$n-i$} & $2n + i + \Gamma + 2\outdeg(c) + \indeg(c) + 2$ & & &$n - i$ & $2n + i + \Gamma + 3$ \\
 \hline
 \multicolumn{1}{|c}{7.} & \multicolumn{1}{||c|}{$n-i$} & $3n + \Gamma + 2\outdeg(c) + \indeg(c) + 3$ & & &$n - i$ & $3n + \Gamma + 4$ \\
 \hline
 \multicolumn{1}{|c}{8.} & \multicolumn{1}{||c|}{$i$} & $4n - i + \Gamma + 2\outdeg(c) + 2\indeg(c) + 3$ & & &$i$ & $4n - i + \Gamma + 7$ \\
 \hline
 \multicolumn{1}{|c}{9.} & \multicolumn{1}{||c|}{$n^3$} & $4n + \Gamma + 10$ & & & $n^3$ & $4n + \Gamma + 8$\\
 \hline
\end{tabular}
\end{center}

By $\deg$ in the  selector vertices threshold is meant the degree of selector vertices not the degree of $c \in V(H)$.
If $c \not \in V^r$, then the vertex gadget has smaller number of groups than in the general case.
Note that $i,j < n$ thus every group is non-empty.
We denote all vertices in all cliques $C_u$ for all $u \in V_c$ by $U_c$.
Observe that all of the cliques in~$U_c$ are of the same size $n^3 + 4n + \gamma_s + 1$, independent on the value of $i$.

\subparagraph{Edge Gadget}
Let $e' = (c,d) \in E(H)$.
Again, we first introduce few constants.
Let $\delta_s = 17$, $\delta_\ell = 8$, and $\Delta = \delta_s + \delta_\ell + 1$.
For each edge $e = \{u,v\} \in E_{cd}$ we create a clique $C_{e}$.
Let $u = v^c_i$ and $v = v^d_j$.
There are again eleven groups of vertices in $C_{e}$.
The sizes of the groups and thresholds of vertices in the group are in the following table.

\begin{center}
\begin{tabular}{|c|c|c|}
 \hline
 Name & Size & Threshold \\
 \hline
 selector & $\delta_s$ & $\deg - \delta_s + 1$ \\
 \hline
 starter & $1$ & $\delta_s$ \\
 \hline
 1. & $n - i$ & $\Delta$ \\
 \hline
 2. & $i$ & $n - i + \Delta + 1$ \\
 \hline
 3. & $i$ & $n + \Delta + 2$ \\
 \hline
 4. & $n - i$ & $n + i + \Delta + 3$ \\
 \hline
 5. & $n - j$ & $2n + \Delta + 4$ \\
 \hline
 6. & $j$ & $3n - j + \Delta + 5$ \\
 \hline
 7. & $j$ & $3n + \Delta + 6$ \\
 \hline
 8. & $n - j$ & $3n + j + \Delta + 7$ \\
 \hline
 9. & $n^4$ & $4n + \Delta + 8$ \\
 \hline
\end{tabular}
\end{center}
All vertices of these cliques are denoted by $U_{e'}$.
Observe that all of the cliques in~$U_{e'}$ are of the same size~$n^4 + 4n+\delta_s+1$, independent on the value of~$i$ or~$j$.

We set the budget $b$ to $\gamma_s \cdot k + \delta_s \cdot \bigl|E(H)\bigr|$.
This will allow us to select one vertex (by adding all of its selector vertices) from each color class $V_c$ as well as one edge from each edge set $E_{cd}$.

\subparagraph{Twin Cover Vertices}
So far our graph is a cluster graph and thus admits a twin cover of size $0$.
Now, we are going to add few vertices (forming the twin cover of the constructed graph) in order to connect these pieces together and the reduction to work.
For each edge $e' = (c,d) \in E(H)$ we add four \emph{checker} vertices defined in the following table.

\begin{center}
\begin{tabular}{|c|c|c|}
\hline
Name & Symbol & Threshold \\
\hline
Outgoing lower checker & $\ell^{-}(c,e')$ & $n + \gamma_s + \delta_s + 2$ \\
\hline
Outgoing upper checker & $u^{-}(c,e')$ & $3n + \gamma_s + \delta_s + 2$ \\
\hline
Ingoing lower checker & $\ell^+(d,e')$ &  $5n + \gamma_s + \delta_s + 2$ \\
\hline
Ingoing upper checker & $u^+(d,e')$ & $7n + \gamma_s + \delta_s + 2$ \\
\hline
\end{tabular}
\end{center}

The outgoing checkers $\ell^-(c,e')$ and $u^-(c,e')$ are both connected to all vertices in $U_c$ and $U_{e'}$.
Similarly, the ingoing checkers $\ell^+(d,e')$ and $u^+(d,e')$ are both connected to all vertices in $U_d$ and $U_{e'}$.
Note that for each vertex $c \in V(H)$ there are $2\cdot\outdeg(c)$ outgoing checkers and $2\cdot\indeg(c)$ ingoing checkers connected to $U_c$.
The checker vertices will check if the selection of vertices in $V_c$ and $V_d$ is consistent with the selection of edge in $E_{cd}$.
For each type of checkers we add one checker validator vertex defined in the following table.
\begin{center}
\begin{tabular}{|c|c|c|}
\hline
Name & Symbol & Threshold \\
\hline
Outgoing lower checker validator & $\tilde{\ell}^{-}$ & $k\cdot(n + \gamma_s + 1) + \bigl|E(H)\bigr|\cdot(n + \delta_s + 2)$ \\
\hline
Outgoing upper checker  validator & $\tilde{u}^{-}$ & $k\cdot(2n + \gamma_s + 1) + \bigl|E(H)\bigr|\cdot(2n + \delta_s + 2)$ \\
\hline
Ingoing lower checker  validator & $\tilde{\ell}^+$ &  $k\cdot(3n + \gamma_s + 1) + \bigl|E(H)\bigr|\cdot(3n + \delta_s + 2)$ \\
\hline
Ingoing upper checker  validator & $\tilde{u}^+$ & $k\cdot(4n + \gamma_s + 1) + \bigl|E(H)\bigr|\cdot(4n + \delta_s + 2)$ \\
\hline
\end{tabular}
\end{center}

All four checker validators are connected to all twin cliques (i.e., the vertices in $U_c$ and $U_{e'}$ for each $c \in V(H)$ and $e' \in E(H)$).
Further, each checker validator is connected to all checkers of corresponding type, i.e., $\tilde{\ell}^-$ is connected to all outgoing lower checkers $\ell^-(c,e')$, $\tilde{u}^-$ is connected to all outgoing upper checkers $u^-(c,e')$ etc.
The purpose of the checker validators is that if some checker is not activated due to an invalid selection, then the corresponding checker validator is not activated as well and it will stop the activation process in the whole graph $G'$.

Next, for each $c \in V(H)$ we add~$\gamma_\ell$ \emph{class validator} vertices $L_c$ and~$\gamma_s$ \emph{sentry} vertices $R_c$.
The class validator vertices $L_c$ and the sentry vertices $R_c$ are connected to all vertices in $U_c$.
The threshold of class validator vertices $L_c$ is~$\gamma_s+1$ and the threshold of sentry vertices $R_c$ is set to the clique size in~$U_c$, i.e., $n^3 + 4n + \gamma_s + 1$.

Similarly, for each $e' \in E(H)$ we add~$\delta_\ell$ class validator vertices $L_{e'}$ and $\delta_s$ sentry vertices $R_{e'}$.
The class validator vertices $L_{e'}$ and the sentry vertices $R_{e'}$ are connected to all vertices in $U_{e'}$.
The threshold of class validator vertices $L_{e'}$ is~$\delta_s+1$ and the threshold of sentry vertices $R_{e'}$ is set to the clique size in~$U_{e'}$, i.e., $n^4 + 4n + \delta_s + 1$.

The purpose of class validator vertices is that vertices only from one clique $C_v$ in each $U_c$ are selected to the target set.
The intended meaning of sentry vertices is to be activated when for each $c \in V(H)$ one vertex gadget $C_v$ in $U_c$ is activated and for each $e' \in E(H)$ one edge gadget $C_{e}$ in $U_{e'}$ is activated.
After that the sentry vertices will help to activate the remaining (i.e., not selected) gadgets.

Let $\tau$ denote the number of non-selector vertices constructed so far.
We finish the construction by adding $b+\gamma_s$ \emph{special} vertices with threshold $\tau + b$ and we connect these to every vertex constructed so far.
These vertices play an important role in our reduction as they ensure that any target set consists of selector vertices only (see Lemma~\ref{lem:TCTargetSetContainsOnlyGuards}).
Recall that $b$ is linear in $k$.
Thus, the twin cover of the constructed graph has size
\[
  \tc(G') \leq (\gamma_s + \gamma_\ell)\cdot k + (4 + \delta_\ell + \delta_s)\cdot\bigl|E(H)\bigr| + b + \gamma_s + 4 = \BigO(k) \,.
\]

\subsubsection{Subgraph Gives Target Set}
Suppose $G$ contains a colored copy of $H$ given by a mapping $\phi\colon V(H) \to V(G)$.
We put the following selector vertices to the set $S$.
\begin{itemize}
  \item
    Let $c \in V(H)$ and $v = \phi(c) \in V(G)$.
    We put all~$\gamma_s$ selector vertices from $C_v$ into $S$.
  \item
    Let $(c,d) \in E(H)$ and $u = \phi(c), v = \phi(d)$.
    By property of $\phi$ we know that $e = \{u,v\} \in E(G)$.
    We put all~$\delta_s$ selector vertices from $C_e$ into $S$.
\end{itemize}
It is clear that $|S| = b$.
We claim that $S$ forms a target set for the graph $G'$ constructed in the previous section.
The high-level description of the activation process is that first the vertices in $C_v$ and $C_e$ cliques with checkers and checker validators are activated for the selected $v \in V(G)$ and $e \in E(G)$.
Then, the sentry vertices are activated which causes activation of all vertex and edge gadgets without the selector vertices.
At the end the special vertices and remaining selector vertices are activated.

Let $V^* = \bigl\{ \phi(c) | c \in V(H)\bigr\}$ and let $E^* = \bigl\{ \{\phi(c),\phi(d)\} | (c,d) \in E(H) \bigr\}$ be the sets vertices and edges of $G$ selected by $\phi$.
Let $\bar{u} \in V_c \setminus V^*$ and $\bar{e} \in E_{c'd'} \setminus E^*$ and $e' = (c',d') \in E(H)$.
Note that for the vertex $\bar{u}$ and the edge the edge $\bar{e}$ the minimal threshold in $C_{\bar{u}}$ is~$\gamma_s$ and in $C_{\bar{e}}$ it is~$\delta_s$.
Thus, some vertices in these cliques can be activated after at least~$\gamma_s$ (or~$\delta_s$) vertices in their neighborhood in twin cover are activated.
We will show that these neighborhood vertices of $C_{\bar{u}}$ (or $C_{\bar{e}}$) are activated in the order:
\begin{enumerate}
\item $\gamma_\ell$ (or $\delta_\ell$) class validator vertices $L_c$ (or $L_{e'}$).
\item 6 (or 4) checkers and 4 checker validators.
\item $\gamma_s$ (or $\gamma_s$) sentry vertices $R_c$ (or $R_{e'}$).
\item $b + \gamma_s$ special vertices.
\end{enumerate}
Since $21 = \gamma_s >  10 + \gamma_\ell = 20$, the first vertex in $C_{\bar{u}}$ is activated after the sentry vertices $R_c$ are activated.
By similar argument, the same holds for $C_{\bar{e}}$, i.e., the first vertex in $C_{\bar{e}}$ is activated after the sentry vertices $R_{e'}$.

Let $e' = (c, d) \in E(H), u = v^c_i \in V^*, v = v^d_j \in V^*$, and $e = \{u,v\} \in E^*$, i.e., $u$ and $v$ are the selected vertices from $V_c$ and $V_d$ and $e$ is the selected edge from $E_{cd}$ incident to $u$ and $v$.
We describe the activation process from the point of view of $u, v$, and $e$, i.e., the following rounds holds for every $u^* \in V^*, \hat{c} \in V(H), e^* \in E^*$, and $\hat{e} \in E(H)$, not only for $u,v,c,d,e$, and $e'$.
In what follows please refer to Figure~\ref{fig:TCTSSOverview}.
\begin{figure}[bt]
  \usetikzlibrary{calc,fit}
\begin{tikzpicture}[node distance=.7cm]
  \tikzstyle{part}=[draw,minimum width=.8cm,minimum height=.7cm]
  \tikzstyle{fitterSel}=[draw,line width=3pt,inner sep=1.5pt]
  \tikzstyle{fitter}=[draw,thick,dotted,inner sep=2pt]
  \tikzstyle{selected}=[]
  \tikzstyle{guardTCVertex}=[fill=yellow!40]
  \tikzstyle{sentryTCVertex}=[fill=red!40]
  \tikzstyle{lowercheckerTCVertex}=[fill=violet!40]
  \tikzstyle{uppercheckerTCVertex}=[fill=green!40]

% node w
\begin{scope}
  \node[part,label={180:selector}] (uhat_g) {32};
  \node[part,label={180:starter},below of=uhat_g] (uhat_s) {21};
  \node[part,label={180:1},below of=uhat_s] (uhat_1) {22};
  \node[part,label={180:2},below of=uhat_1] (uhat_2) {23};
  \node[part,label={180:3},below of=uhat_2] (uhat_3) {24};
  \node[part,label={180:4},below of=uhat_3] (uhat_4) {25};
  \node[part,label={180:5},below of=uhat_4] (uhat_5) {26};
  \node[part,label={180:6},below of=uhat_5] (uhat_6) {27};
  \node[part,label={180:7},below of=uhat_6] (uhat_7) {28};
  \node[part,label={180:8},below of=uhat_7] (uhat_8) {29};
  \node[part,label={180:9},below of=uhat_8] (uhat_9) {30};

  \node[fitter,fit=(uhat_g)(uhat_9),label={270:$C_{\hat{u}}$}] {};

  \node[part,selected] (u_g) at (1.5,0) {0};
  \node[part,below of=u_g] (u_s) {1};
  \node[part,below of=u_s] (u_1) {3};
  \node[part,below of=u_1] (u_2) {5};
  \node[part,below of=u_2] (u_3) {7};
  \node[part,below of=u_3] (u_4) {9};
  \node[part,below of=u_4] (u_5) {11};
  \node[part,below of=u_5] (u_6) {13};
  \node[part,below of=u_6] (u_7) {15};
  \node[part,below of=u_7] (u_8) {17};
  \node[part,below of=u_8] (u_9) {19};

  \node[fitterSel,fit=(u_g)(u_9),label={270:$C_u$}] {};

  \node at ($(u_g)!.5!(uhat_g) + (0,.7)$) {$U_c$};

  %class guard
  \node[part,guardTCVertex,label={0:$L_c$}] at ($(uhat_9)!.5!(u_9) - (0,1.2)$) (w_guard) {2};
  \draw (uhat_9) to[out=-90,in=90] (w_guard);
  \draw (u_9) to[out=-90,in=90] (w_guard);

  %class sentry
  \node[part,sentryTCVertex,label={270:$R_c$}] at ($(uhat_9)!.5!(u_9) - (0,2.5)$) (w_sentry) {20};
  \draw (uhat_9) to[out=-140,in=180] (w_sentry);
  \draw (u_9) to[out=-40,in=0] (w_sentry);
\end{scope}

%node x
\begin{scope}[xshift=8cm]
  \node[part] (vhat_g) at (1.5,0) {32};
  \node[part,below of=vhat_g] (vhat_s) {21};
  \node[part,below of=vhat_s] (vhat_1) {22};
  \node[part,below of=vhat_1] (vhat_2) {23};
  \node[part,below of=vhat_2] (vhat_3) {24};
  \node[part,below of=vhat_3] (vhat_4) {25};
  \node[part,below of=vhat_4] (vhat_5) {26};
  \node[part,below of=vhat_5] (vhat_6) {27};
  \node[part,below of=vhat_6] (vhat_7) {28};
  \node[part,below of=vhat_7] (vhat_8) {29};
  \node[part,below of=vhat_8] (vhat_9) {30};

  \node[fitter,fit=(vhat_g)(vhat_9),label={270:$C_{\hat{v}}$}] {};

  \node[part,selected] (v_g) {0};
  \node[part,below of=v_g] (v_s) {1};
  \node[part,below of=v_s] (v_1) {3};
  \node[part,below of=v_1] (v_2) {5};
  \node[part,below of=v_2] (v_3) {7};
  \node[part,below of=v_3] (v_4) {9};
  \node[part,below of=v_4] (v_5) {11};
  \node[part,below of=v_5] (v_6) {13};
  \node[part,below of=v_6] (v_7) {15};
  \node[part,below of=v_7] (v_8) {17};
  \node[part,below of=v_8] (v_9) {19};

  \node[fitterSel,fit=(v_g)(v_9),label={270:$C_v$}] {};

  \node at ($(v_g)!.5!(vhat_g) + (0,.7)$) {$U_d$};

  %global guard
  \node[part,guardTCVertex,label={0:$L_d$}] at ($(vhat_9)!.5!(v_9) - (0,1.2)$) (x_guard) {2};
  \draw (vhat_9) to[out=-90,in=90] (x_guard);
  \draw (v_9) to[out=-90,in=90] (x_guard);

  %class sentry
  \node[part,sentryTCVertex,label={270:$R_d$}] at ($(vhat_9)!.5!(v_9) - (0,2.5)$) (x_sentry) {20};
  \draw (v_9) to[out=-140,in=180] (x_sentry);
  \draw (vhat_9) to[out=-40,in=0] (x_sentry);
\end{scope}

%edge (w,x)
\begin{scope}[xshift=4cm]
  \node[part] (ehat_g) at (1.5,0) {32};
  \node[part,below of=ehat_g] (ehat_s) {21};
  \node[part,below of=ehat_s] (ehat_1) {22};
  \node[part,below of=ehat_1] (ehat_2) {23};
  \node[part,below of=ehat_2] (ehat_3) {24};
  \node[part,below of=ehat_3] (ehat_4) {25};
  \node[part,below of=ehat_4] (ehat_5) {26};
  \node[part,below of=ehat_5] (ehat_6) {27};
  \node[part,below of=ehat_6] (ehat_7) {28};
  \node[part,below of=ehat_7] (ehat_8) {29};
  \node[part,below of=ehat_8] (ehat_9) {30};

  \node[fitter,fit=(ehat_g)(ehat_9),label={270:$C_{\{\hat{u},\hat{v}\}}$}] {};

  \node[part,selected] (e_g) {0};
  \node[part,below of=e_g] (e_s) {1};
  \node[part,below of=e_s] (e_1) {3 };
  \node[part,below of=e_1] (e_2) {5};
  \node[part,below of=e_2] (e_3) {7};
  \node[part,below of=e_3] (e_4) {9};
  \node[part,below of=e_4] (e_5) {11};
  \node[part,below of=e_5] (e_6) {13};
  \node[part,below of=e_6] (e_7) {15};
  \node[part,below of=e_7] (e_8) {17};
  \node[part,below of=e_8] (e_9) {19};

  \node[fitterSel,fit=(e_g)(e_9),label={270:$C_{\{u,v\}}$}] {};

  \node at ($(e_g)!.5!(ehat_g) + (0,.7)$) {$U_{e'}$};

  %global guard
  \node[part,guardTCVertex,label={0:$L_{e'}$}] at ($(ehat_9)!.5!(e_9) - (0,1.4)$) (wx_guard) {2};
  \draw (ehat_9) to[out=-130,in=90] (wx_guard);
  \draw (e_9) to[out=-50,in=90] (wx_guard);

  %class sentry
  \node[part,sentryTCVertex,label={270:$R_{e'}$}] at ($(ehat_9)!.5!(e_9) - (0,2.5)$) (e_sentry) {20};
  \draw (e_9) to[out=-140,in=180] (e_sentry);
  \draw (ehat_9) to[out=-40,in=0] (e_sentry);
\end{scope}

% lower&upper w->out
\begin{scope}[xshift=2.75cm,yshift=1.8cm,node distance=.8cm]
  \node[part,lowercheckerTCVertex,label={[yshift=2pt]270:$\ell^-(c,e')$}] (w_lower) {4};
  \node[part,above of=w_lower,uppercheckerTCVertex,label={[yshift=-2pt]90:$u^-(c,e')$}] (w_upper) {8};
  %\node[above of=w_upper,text width=1.3cm] {outgoing checkers};
\end{scope}

% lower&upper x->in
\begin{scope}[xshift=6.75cm,yshift=1.8cm,node distance=.8cm]
  \node[part,lowercheckerTCVertex,label={[yshift=2pt]270:$\ell^+(d,e')$}] (x_lower) {12};
  \node[part,above of=x_lower,uppercheckerTCVertex,label={[yshift=-2pt]90:$u^+(d,e')$}] (x_upper) {16};
  %\node[above of=x_upper,text width=1.3cm] {ingoing checkers};
\end{scope}

% U_
\begin{scope}[on background layer]
  \coordinate (cliqueFitterNW) at ($ (uhat_g.north west) + (-.2,.2) - (1.2,0) + (0,.4)$);
  \draw[rounded corners,gray,fill=gray!20] (cliqueFitterNW) rectangle ($ (vhat_9.south east) + (.2,-.15)$);
\end{scope}

\node[part,label={$ \tilde{\ell}^-$}] (outLowerValidator) at ($ (uhat_g.north) - (1,0) + (0,1.5) $) {6};
\node[part,label={$ \tilde{u}^-$}] (outUpperValidator) at ($ (outLowerValidator) + (1.3,0) + (0,.7) $) {10};

\coordinate (cliqueFitterNE) at ($(cliqueFitterNW) + (11.5,0)$);
\node[part,label={$\tilde{\ell}^+$}] (inLowerValidator) at ($ (vhat_g.north) + (.7,1.5) $) {14};
\node[part,label={$\tilde{u}^+$}] (inUpperValidator) at ($ (vhat_g.north) + (0,2.2) - (.3,0) $) {18};

%global edges
%w->V_w
\draw (uhat_g) to[out=90,in=180] (w_lower);
\draw (uhat_g) to[out=90,in=180] (w_upper);
\draw (u_g) to[out=90,in=180] (w_lower);
\draw (u_g) to[out=90,in=180] (w_upper);
%w->E_{(w,x)}
\draw (ehat_g) to[out=90,in=0] (w_lower);
\draw (ehat_g) to[out=90,in=0] (w_upper);
\draw (e_g) to[out=90,in=0] (w_lower);
\draw (e_g) to[out=90,in=0] (w_upper);
%x->V_x
\draw (vhat_g) to[out=90,in=0] (x_lower);
\draw (vhat_g) to[out=90,in=0] (x_upper);
\draw (v_g) to[out=90,in=0] (x_lower);
\draw (v_g) to[out=90,in=0] (x_upper);
%x->E_{(w,x)}
\draw (ehat_g) to[out=90,in=180] (x_lower);
\draw (ehat_g) to[out=90,in=180] (x_upper);
\draw (e_g) to[out=90,in=180] (x_lower);
\draw (e_g) to[out=90,in=180] (x_upper);

%%%
\draw (outLowerValidator) to ($ (cliqueFitterNW) + (1,0) $);
\draw (outLowerValidator) to (w_lower);
\draw (outUpperValidator) to ($ (cliqueFitterNW) + (1,0) $);
\draw (outUpperValidator) to (w_upper);

\draw (inLowerValidator) to ($ (cliqueFitterNE) $);
\draw (inLowerValidator) to (x_lower);
\draw (inUpperValidator) to ($ (cliqueFitterNE) $);
\draw (inUpperValidator) to (x_upper);

%delimiting cliques
\draw ($(w_lower) - (0,.8)$) to ($(w_lower) - (0,11.6)$);
\draw ($(x_lower) - (0,.8)$) to ($(x_lower) - (0,11.6)$);

\end{tikzpicture}
  \caption{\label{fig:TCTSSOverview}%
  A schema of an activation process for a Yes-instance of~CSI.
  If there is an edge between two groups of vertices, this indicates that every vertex of one group is adjacent to every vertex in the other group.
  A rectangle represents a group of vertices and a number~$\ell$ stands for activation in the round~$S_\ell$.
  Vertices in the twin cover have the following colors: class validators yellow, lower checkers violet, upper checkers green, class sentry vertices light red, and checker validators white (special vertices are omitted).
  %The vertex marked OLCV is the outgoing lower checker validator while the one marker OUCV is the outgoing upper checker validator; similarly for the ingoing checker validators ILCV and IUCV.
  Groups in vertex and edge gadgets are in the same order as in the table.
  For simplicity we only depict two vertex gadgets from $\phi^{-1}(c)$ and $\phi^{-1}(d)$ with $\deg^-(c),\deg^-(d) \in \{1,2\}$; the selected vertex gadget has bold borders.
  %The edges stand for complete bipartite graphs between the two groups of vertices.
  Finally, the gray area ``contains''all the twin-cliques (the vertices in $U_c$ and $U_{e'}$ for each $c \in V(H)$ and $e' \in E(H)$).
  }
\end{figure}

\begin{description}
 \item[$S_1$: Starters.]
   In $C_u$ and $C_v$ the starter vertex is activated because there are~$\gamma_s$ active selector vertices in $C_u$ and in $C_v$.
   Similarly, the starter vertex in $C_e$ is activated.
 \item[$S_2$: Class validators.]
   The class validator vertices $L_c$ and $L_d$ are activated because there are~$\gamma_s+1$ active vertices in $C_u$ and $C_v$.
   Similarly, the class validator vertices $L_{e'}$ are activated.
 \item[$S_3$: First groups.]
   In $C_u$ the first groups are activated because there are~$\gamma_s+1$ active vertices in $C_u$ (selector and starter) and~$\gamma_\ell$ active class validators in the neighborhood of $C_u$.
   Similarly, the first group in $C_e$ and in $C_v$ is activated.
 \item[$S_4$: Outgoing lower checkers.]
   The lower checker $\ell(c, e')$ of threshold $n + \gamma_s + \delta_s + 2$ is activated.
   It has $i + \gamma_s + 1$ active neighbors in $C_u$ and $n - i + \delta_s + 1$ active neighbors in $C_e$.
 \item[$S_5$: Second groups.]
   There are $\outdeg(c)$ active outgoing lower checkers in the neighborhood of~$C_u$ and $\outdeg(d)$ active outgoing lower checkers in the neighborhood of~$C_v$ (and~$\gamma_\ell$ active class validators $L_c$ and $L_d$).
   Thus, the second group in $C_u$ is activated, since in total there are
   \[
     i + \gamma_s + 1 + \outdeg(c) + \gamma_\ell = i + \Gamma + \outdeg(c)
   \]
   active vertices in their neighborhood.
   Similarly, if $\outdeg(d) \neq 0$, the second group in $C_v$ is activated as well.
   There is one active outgoing lower checker in the neighborhood of~$C_e$, thus the second group of~$C_e$ is activated.
 \item[$S_6$: Outgoing lower checker validator.]
   For each $\hat{c} \in V(H)$ there is exactly one clique $C \in U_{\hat{c}}$ such that there are $n + \gamma_s + 1$ active vertices in $C$ (selector, starter and first two groups).
   Similarly, for each $\hat{e} \in E(H)$ there is exactly one clique $C' \in U_{\hat{e}}$ such that there are $n + \delta_s + 1$ active vertices $C'$.
   Further, for each $\hat{e} = (\hat{c}, \hat{d}) \in E(H)$ the outgoing lower checker $\ell^{-}(\hat{c}, \hat{e})$ is active.
   Thus, the outgoing lower checker validator $\tilde{\ell}^-$ is activated as it has
   \[
    k\cdot (n + \gamma_s + 1) + \bigl|E(H)\bigr|\cdot (n + \delta_s + 2)
   \]
   active neighbors.
 \item[$S_7$: Third groups.]
   There are $\outdeg(c)$ active lower checkers and the active outgoing lower checker validator $\tilde{\ell}^-$ in the neighborhood of~$C_u$.
   Thus, the third group of threshold $n + \Gamma + \outdeg(c) + 1$ in~$C_u$ is activated.
   Similarly, if $\outdeg(d) \neq 0$, the third group in~$C_v$ is activated.
   There is one active lower checker and $\tilde{\ell}^-$ in the neighborhood of $C_e$, thus the third group of $C_e$ is activated as well.
 \item[$S_8$: Outgoing upper checkers.]
   The outgoing upper checker $u^-(c, e')$ of threshold $3n + \gamma_s + \delta_s + 2$ is activated.
   It has $2n - i + \gamma_s + 1$ active neighbors in $C_u$ and $n + i + \delta_s + 1$ active neighbors in $C_e$.
 \item[$S_9$: Fourth groups.]
   There are $\outdeg(c)$ active outgoing lower checkers, $\outdeg(c)$ active outgoing upper checkers and active $\tilde{\ell}^-$ in the neighborhood of~$C_u$.
   Thus, the fourth group in~$C_u$ is activated.
   Similarly, if $\outdeg(d) \neq 0$, the fourth group in~$C_v$ is activated.
   There are two active outgoing checkers and active $\tilde{\ell}^-$ in the neighborhood of $C_e$, thus the fourth group of $C_e$ is activated as well.
 \item[$S_{10}$: Outgoing upper checker validator.]
   The vertex $\tilde{u}^-$ has $2n + \gamma_s + 1$ active neighbors in each $U_{\hat{c}}$ (there is one clique with active selector, starter and first four groups) and $2n + \gamma_s + 1$ active neighbors in each $U_{\hat{e}}$.
   Further, there are $|E(H)|$ active outgoing upper checkers in the neighborhood of $\tilde{u}^-$.
   Thus, the vertex $\tilde{u}^-$ is activated.
 \item[$S_{11}$--$S_{18}$: Groups 5--8, ingoing checkers and ingoing checker validators.]
   In the next 8 rounds ingoing checkers and corresponding checker validators with the groups 5--8 are activated in the similar manner as the group 1--4 with outgoing checkers and checker validators.
   Note that if $c \in V^-$, then in the round 11, the whole clique $C_u$ is activated and in the round 12 the sentry vertices $R_c$ are activated.
   Furthermore, other non-selector vertices in $U_c$ are activated in subsequent rounds 13--18 (for the description of similar process for $c \in V^r$ see rounds~21--30 below).
   Then, nothing is activated in such classes until the round~32.
 \item[$S_{19}$: Ninth groups.]
   There are all~6 active checkers and all 4 active checker validators in the neighborhood of $C_u$, thus the ninth group is activated (if $c \not \in V^-$) and same for $C_v$.
   There are all~4 active checkers all 4 active checker validators in the neighborhood of $C_e$, thus the ninth group of $C_e$ is activated.
 \item[$S_{20}$: Sentries.]
   Up until now we have activated the following vertices:
   \begin{itemize}
     \item The vertex gadget $C_v$ of size $n^3 + 4n + \gamma_s + 1$ and the vertex gadget $C_u$.
     \item The edge gadget $C_e$ of size $n^4 + 4n + \delta_s + 1$.
     \item All checkers, checker validators and class validators.
   \end{itemize}
   Every sentry vertex $r \in R_c$ has $n^3 + 4n + \gamma_s + 1$ active neighbors -- namely, the clique $C_v$.
   Thus, all sentry vertices in $R_c$ are activated.
   Similarly, all sentry vertices in $R_d$ and $R_{e'}$ are activated.
 \item[$S_{21}$--$S_{30}$: Non-selector vertices.]
   Every vertex gadget has~$10 + \gamma_s + \gamma_\ell$ active neighbors (6~checkers, 4~checker validators, $\gamma_\ell$ class validators and $\gamma_s$ sentry vertices).
   Thus, every vertex gadget in $U_c$ (if $c \not \in V^-$) without the selector vertices is activated group by group in~10~rounds (see Figure~\ref{fig:TCTSSOverview}).
   The sentry vertices substitute the selector vertices during this phase.
   The non-selector vertices in $U_d$ are activated in the similar way.
   If $d \in V^+$ then non-selector vertices in $U_d$ are activated in 6 rounds (and then no vertex is activated in $U_d$ till the round $S_{32}$).
   Every edge gadget in $U_{e'}$ has~$8 + \delta_s + \delta_\ell$ active neighbors and is activated in the same way as each vertex gadget.
 \item[$S_{31}$: Special vertices.]
   All $\tau$ non-selector and non-special vertices are active.
   Since $b$ selector vertices are active from the beginning, the special vertices are activated.
 \item[$S_{32}$: Guard.]
   All vertices except selector vertices not in $S$ are active.
   Note that, each selector vertex has threshold its degree minus the number of selector vertices in its neighborhood.
   Thus, all remaining selector vertices are activated and $S_{32} = V(G')$.
\end{description}

Thus, we prove the following theorem.

\begin{theorem}
\label{thm:SubgraphTSS}
 If $G$ contains a colored copy of $H$, then $G'$ with the threshold function $f$ contains a target set of size $b$.
 \qed
\end{theorem}

\subsubsection{Target Set Gives Subgraph}
In this section we want to prove the converse implication of Theorem~\ref{thm:SubgraphTSS}.
Let $S$ be a target set of $G'$ and $|S| \leq b$.
First, we want to prove that basically the only possibility is that $S$ corresponds to a (valid) selection of vertices and edges (i.e., $S$ contains all selector vertices of exactly one vertex or edge gadget in each class).
We prove this in the following lemmata.
\begin{lemma}\label{lem:TCTargetSetContainsOnlyGuards}
The set $S$ consists of $b$ selector vertices.
\end{lemma}
\begin{proof}
Suppose the set $S$ does not contain $b$ selector vertices.
Note that there are at least~$\gamma_s$ special vertices not in $S$.
Let $T$ be the set of all special vertices not in~$S$.
We show that $S$ cannot activate vertices in $T$.
Let $v$ be a selector vertex which is not in $S$.
By our construction, the vertex $v$ has at most~$\gamma_s$ other selector vertices as neighbors.
Since $f(v) \geq \deg(v) - \gamma_s + 1$ (note that $f(v) = \deg(v) - \gamma_s + 1$ for a selector vertex in a vertex gadget and $f(v) = \deg(v) - \delta_s + 1$ for a selector vertex in an edge gadget; furthermore, $\gamma_s \ge \delta_s$), the vertex $v$ is activated after at least one vertex in $T$ is activated.

Thus, any selector vertex which is not in $S$ can be activated only after at least one vertex in $T$ activated.
However, the vertices in $T$ have threshold $\tau + b$.
Therefore, even if all non-selector vertices are activated, at least $b$ selector vertices have to be in $S$.
\end{proof}

Before we prove that in each set $U_z$ (for $z \in V(H) \cup E(H)$) there is a correct number of selector vertices in~$S$ we first observe few facts about the last rounds of the activation process arising from a set of selector vertices~$S$.
In order to give more properties (in fact, more ``local'' ones) the set~$S$ must fulfill it helps us if we understand the last rounds of the activation process as it mainly deals with the globally connected vertices forming the twin cover.
This allows us to prove the sought properties more in a local way.
\begin{lemma}\label{lem:TCHardness:lastRoundsOfTheActivation}
  The following holds in the activation process arising from~$S$:
  \begin{enumerate}
    \item Selector vertices not in~$S$ are activated in the last round of the activation process.
    \item Special vertices are activated in the last but one round.
  \end{enumerate}
\end{lemma}
\begin{proof}
  The threshold of special vertices is~$\tau + b$, thus they can be activated when all $\tau$ non-selector vertices are activated.
  The selector vertices have their threshold set exactly to the number of their non-selector neighbors and they are connected to the special vertices.
  Thus, selectors not in $S$ have to be activated after the special vertices.
  We conclude the activation of these two groups of vertices takes place in the last two rounds of the activation process arising from~$S$.
\end{proof}

In particular, the above lemma implies that even if a single (group of) sentry vertices is not activated before the special vertices and the selector vertices not in $S$, the activation process arising from~$S$ terminates before all of the vertices are active and thus~$S$ is not a target set.
Let $z \in V(H) \cup E(H)$.
For a clique $C$ in $U_z$ we denote the non-empty group with the highest number as the \emph{last group} (for $z \in V^-$ the last group is the fifth group and for $z \not \in V^-$ the last group is the ninth group).
\begin{lemma}
\label{lem:SentryActivation}
 Let $z \in V(H) \cup E(H)$.
 In the activation process arising from~$S$, the sentry vertices in $R_z$ can be activated after the last group of at least one clique in $U_z$ is activated.
\end{lemma}
\begin{proof}
 Suppose $z \in V(H)$.
 Then, the threshold of vertices in $R_z$ is larger than $n^3$.
 The size of each clique $C_v$ in $U_z$ without the last group is $4n + \gamma_s + 1$ and there are $n - 1$ cliques in $U_z$.
 Thus, even if all cliques in $U_z$ without the last groups are activated, the sentry vertices in $R_z$ can not be activated, since $n^3 > (4n + \gamma_s + 1)(n-1)$ for $n \geq 6$ and the only neighborhood of $R_z$ are the special vertices and vertices in $U_z$.

 The proof for $z \in E(H)$ is analogous.
 There are at most $n^2$ cliques in $U_z$ and each cliques without the last group has $4n + \delta_s + 1$ vertices.
 Since the thresholds of sentry vertices in $R_z$ is at least $n^4$, the lemma follows.
\end{proof}

Now, we are ready to finish the first part of the proof -- namely, that the selection indeed corresponds to a selection of vertices and edges in the original instance of~CSI.

\begin{lemma}
\label{lem:GuardsOnlyForOneObject}
 For each $c \in V(H)$ there is exactly one vertex $v \in V_c$ such that $V(C_v) \cap S$ contains exactly $\gamma_s$ selectors in the vertex gadget $C_v$.
 For each $(c,d) \in E(H)$ there is exactly one edge $e \in E_{cd}$ such that $V(C_e) \cap S$ contains exactly $\delta_s$ selectors in the edge gadget $C_e$.
\end{lemma}
\begin{proof}
 We prove the lemma for edges $E(H)$ (the proof for vertices $c \in V(H)$ follows in a similar way but we have to distinguish cases if $c$ is in $V^-$ or not).
Suppose there is $e' = (c,d) \in E(H)$ such that for all $e \in E_{cd}$ holds  $|V(C_e) \cap S| < \delta_s$.
 We prove the sentry vertices $R_{e'}$ are never activated.
 Let us assume all~4~checker, all 4 checker validators and the class validator $L_{e'}$ are activated.
 Thus, in each $C_e$ all groups but selectors not in $S$ and the last one are activated in 9 rounds.
 However, the last group in $C_e$ can not be activated.
 The last group has the following active neighbors:
 \begin{itemize}
  \item $4n$ vertices in the groups 1--8 in $C_e$.
  \item 1 starter vertex in $C_e$.
  \item Strictly less than $\delta_s$ selectors in $C_e$.
  \item $\delta_\ell$ class validators $L_{e'}$.
  \item 4 checkers and 4 checker validators.
 \end{itemize}
 In total the last group has strictly less than $4n + 1 + \delta_s + \delta_\ell + 8 = 4n + \Delta + 8$, which is the threshold of the last group.
 Thus, the last group in any $C_e$ in $U_{e'}$ can be activated after some sentry vertex in $R_{e'}$, or special vertex or selector vertex in $C_e \setminus S$ is activated.
 However, this cannot happen due to Lemma~\ref{lem:TCHardness:lastRoundsOfTheActivation} and~\ref{lem:SentryActivation}.

 Thus, for each $z \in V(H) \cup E(H)$ there is at least one clique $C \in U_z$ such that $|C \cap S| \geq \delta_s$ (or $\gamma_s$ if $z \in V(H)$).
 Since the size of $b$ is $\gamma_s\cdot k + \delta_s\cdot |E(H)|$, these inequalities are actually tight and the lemma follows.
\end{proof}
By Lemma~\ref{lem:GuardsOnlyForOneObject}, the set $S$ contains:
\begin{enumerate}
 \item For each vertex $c \in V(H)$ exactly~$\gamma_s$ selector vertices in exactly one $C_v$ in $U_c$.
 \item For each edge $e' \in E(H)$ exactly~$\delta_s$ selector vertices in exactly one $C_{e}$ in $U_{e'}$.
\end{enumerate}
Thus, we can define mappings $\phi\colon V(H) \to V(G)$ and $\phi_E\colon E(H) \to E(G)$ such that
\begin{align*}
 &\phi(c) = v \text{ if } v \in V_c\text{ and }|C_v \cap S| = \gamma_s \text{ and } \\
 &\phi_E\bigl((c,d)\bigr) = e \text{ if } e \in E_{cd} \text{ and }|C_e \cap S| = \delta_s \,.
\end{align*}
It remains to prove the mapping $\phi$ is a solution to the instance of CSI.
The next lemma shows that $\phi$ maps edges of $H$ to the edges of $G$.

\begin{lemma}
  For each $(c,d) \in E(H)$ it holds that $\phi_E\bigl((c,d)\bigr) = \bigl\{\phi(c),\phi(d)\bigr\}$.
\end{lemma}
\begin{proof}
  We prove this by a more detailed examination of the activation process arising from the set~$S$.
  Let $e = \phi_E(e')$ for $e' = (c,d) \in E(H)$. % for which we have $u' \in V_c$ and $v' \in V_d$.
  Let~$u = \phi(c)$ and $v = \phi(d)$.
  We will prove that $u,v \in e$.

  We claim the outgoing checkers connected to $C_e$ are activated if and only if $u \in e$.
  Let $u = v^c_i$ for some~$i \in \mathbb{N}$ and let~$e \cap \{v^c_1, \ldots, v^c_{n-1} \} = v^c_j$.
  We will prove that $i = j$.
  We observe that the activation process begins by following the lines of the proof of Theorem~\ref{thm:SubgraphTSS}.
  In the first round the starter vertex in~$C_u$ is activated, then the class validator vertices are activated and then the first group of vertices in~$C_u$ is activated in the third round.
  Thus, there are $i + \gamma_s + 1$ active vertices in~$U_c$ by now.
  Similarly for $e$ -- there are $n-j + \delta_s + 1$ vertices activated in~$C_e$.
  The vertices in the neighborhood of~$C_e$ with the lowest threshold is the outgoing lower checker $\ell^-(c,e')$.

  Now suppose that $n-j+i < n$, i.e., the checker $\ell(c, e')$ can not activated in the round 4.
  Consequently, the second groups of $C_e$ and $C_u$ are not activated in the round 5.
  There can be some cliques $C$ in $U_{z}$ for some $z \in V(H) \cup E(H)$ such that the second group of $C$ is activated in the round 5 (if an appropriate number of outgoing lower checkers connected to $C$ were activated in the round 4).
  We claim in this moment the activation process stops.
  In the round 6 the outgoing lower checker validator $\tilde{\ell}^-$ should be activated.
  However, to activate $\tilde{\ell}^-$ the selector, starter and first two groups of one clique $C$ in $U_z$ for each $z \in V(H) \cup E(H)$ and all outgoing lower checkers have to be active.
  The second group of $C_e$ and $C_u$ and the checker $\ell(c,e')$ are not active, thus $\tilde{\ell}^-$ can not be activated.
  To continue with the activation process, to activate the third group of any gadget clique $C$ the vertex $\tilde{\ell}^-$ needs to be activate.
  Therefore, the activation process stops in the round 5.

  On the other hand, if $i + n-j \ge n$ (i.e., $j \le i$), the lower checker vertex is activated.
  The activation process proceed as it is described in the previous section.
  The outgoing upper checker $u^-(c, e')$ should be activated in the round 8.
  If it is not activated, then the activation process stops latest in the round 9 by the same argument as we presented for the outgoing lower checker.
  Otherwise, if $u^-(c,e')$ is activated, then $2n - i + n + j \geq 3n$ (i.e., $j \geq i$).
  Thus, we prove that if $S$ is a target set, then $u \in e$.

  The argument for~$v \in e$ is similar, only ingoing checker vertices $\ell(d,e')$ and $u(d,e')$ and ingoing checker validators $\tilde{\ell}^+$ and $\tilde{u}^+$ are used.
\end{proof}
This finishes the proof of Theorem~\ref{thm:TSSisWwrtTC}.

\section{Conclusions}
We have generalized ideas of previous works~\cite{BenZwiHLN11,NichterleinNUW13} for the \TSS problem.
The presented results give a new idea how to encode selecting vertices and edges in the {\sc Colored Subgraph Isomorphism} problem for showing {\sf W[1]}-hardness and ETH-lower bound.
In particular, only few problems are known to be \Wh{1} when parameterized by neighborhood diversity -- which is the case for the \TSS problem.
%More importantly, our reductions can be used to obtain further hardness results~\cite{}.
%
We are not aware of other positive results concerning the number of different thresholds instead of the threshold upper-bound.

We would like to point out that in our proofs of \W{1}-hardness the activation process terminates after constant number of rounds (independent of the parameter value and the size of the input graph).
This is true also for all reductions given by Chopin et al.~\cite{ChopinNNW14}.

\bibliography{bib/LP.bib,bib/comsoc.bib,bib/denseGraphs.bib,bib/monographs.bib}

\end{document}